\documentclass[journal]{IEEEtran}
\pdfoutput=1

\ifCLASSINFOpdf
   \usepackage[pdftex]{graphicx}
   \graphicspath{{img/pdf/}{img/jpeg/}}
   \DeclareGraphicsExtensions{.pdf,.jpeg,.png}
\else
   \usepackage[dvips]{graphicx}
   \graphicspath{{img/eps/}}
   \DeclareGraphicsExtensions{.eps}
\fi

\usepackage{flushend}

\usepackage{balance}
\usepackage[cmex10]{amsmath}
\usepackage{booktabs}
\usepackage{amsfonts}
\usepackage{amssymb}
\usepackage{mathrsfs}
\usepackage{cite}
\usepackage{lipsum}
\usepackage{multicol}

\usepackage[tight,footnotesize]{subfigure}
\usepackage{wasysym}
\usepackage{color}
\usepackage{epsfig} % for postscript graphics files
\usepackage{epstopdf}
\usepackage{float}
\usepackage{amsthm}
\usepackage{mathtools}
\allowdisplaybreaks
\newtheorem{lem}{Lemma}

\definecolor{light-gray}{gray}{0.8}

\def\nbb{{\mathbf{b}}}

\def\nbd{{\mathbf{d}}}

\def\nbx{{\mathbf{x}}}
\def\nby{{\mathbf{y}}}
\def\nbz{{\mathbf{z}}}
\def\nb0{{\mathbf{0}}}
\def\nb1{{\mathbf{1}}}

\def\ncalB{{\mathcal{B}}}
\def\ncalC{{\mathcal{C}}}

\def\ncalL{{\mathcal{L}}}

\def\nbbE{{\mathbb{E}}}

\def\nbbP{{\mathbb{P}}}

\def\nbbR{{\mathbb{R}}}

\def\nrmd{{\rm d}}

\def\nrml{{\rm l}}

\def\nrmu{{\rm u}}

\def\nrmx{{\rm x}}
\def\nrmy{{\rm y}}

\def\sinc{{\rm sinc}}

\newtheorem{theorem}{Theorem}
\newtheorem{prop}{Proposition}

\newtheorem{remark}{Remark}

\def\figref#1{Fig.\,\ref{#1}}%
\def\E{\mathbb{E}}

\def\pc{\mathtt{P_c}}
\def\R{\mathbb{R}}

\def\sinr{\mathtt{SINR}}			% Signal to interference plus noise ratio

\def\sir{\mathtt{SIR}}

\def\LDSH{LD-SH}
\def\HDSH{HD-SH}
\def\LDLH{LD-LH}
\def\HDLH{HD-LH}

\def\t {x_\nrmy}
\def\u{x_\nrmx}
\def\yuu{y_{1{\nrmx}}}
\def\yut{y_{1{\nrmy}}}

\def\yuus{y_{2\nrmx}}
\def \yuut{y_{2\nrmy}}

\begin{document}

\title{Poisson Hole Process: Theory and Applications to Wireless Networks}
%\title{Exclusion Zones and Asymmetric Interference Field in Random Wireless Networks}
% make the title area
\author{\IEEEauthorblockN{Zeinab Yazdanshenasan, Harpreet S. Dhillon, Mehrnaz Afshang, 
and Peter Han Joo Chong
}

\thanks{H. S.~Dhillon  and M.~Afshang are with Wireless@VT, Department of ECE, Virgina Tech, Blacksburg, VA, USA. Email: \{hdhillon,mehrnaz\}@vt.edu. Z.~Yazdanshenasan is with Wireless@VT, Department of ECE, Virgina Tech, Blacksburg, VA, USA and with School of EEE, Nanyang Technological University, Singapore. Email: zeinab001@e.ntu.edu.sg. P. H. J.~Chong is with School of EEE, Nanyang Technological University, Singapore. Email: Ehjchong@ntu.edu.sg. This work has been submited in part to the IEEE International Conference on Communications 2016~\cite{YazDhiC2016}. \hfill Manuscript last revised: \today.}}
\maketitle
%\IEEEpeerreviewmaketitle
\begin{abstract}
Interference field in wireless networks is often modeled by a homogeneous Poisson Point Process (PPP). While it is realistic in modeling the inherent node irregularity and provides meaningful first-order results, it falls short in modeling the effect of interference management techniques, which typically introduce some form of {\em spatial interaction} among active transmitters. In some applications, such as cognitive radio and device-to-device networks, this interaction may result in the formation of {\em holes} in an otherwise homogeneous interference field. The resulting interference field can be accurately modeled as a {\em Poisson Hole Process (PHP)}. Despite the importance of PHP in many applications, the exact characterization of interference experienced by a typical node in a PHP is not known. In this paper, we derive several tight upper and lower bounds on the Laplace transform of this interference. 
%For the upper bound, we introduce a new approach in which we {\em dissolve} the holes of a PHP in such as way that it results in a non-homogeneous PPP that provably bounds the interference power of that of a PHP. 
Numerical comparisons reveal that the new bounds outperform all known bounds and approximations, and are remarkably tight in all operational regimes of interest. The key in deriving these tight and yet simple bounds is to capture the {\em local neighborhood} around the typical point accurately while simplifying the far field to attain tractability. Ideas for tightening these bounds further by incorporating the effect of overlaps in the holes are also discussed. These results immediately lead to an accurate characterization of the coverage probability of the typical node in a PHP under Rayleigh fading.

\end{abstract}

%\IEEEpeerreviewmaketitle

\section{Introduction}
%\IEEEPARstart{N}{etwork} topology has great impact on the performance of wireless networks. Indeed, it is the most important parameter for evaluating the signal-to-interference-plus-noise ratio ($\sinr$). This is due to the fact that received signal power and interference depend critically on the locations of transmitters and receivers of the network and especially their relative positions with respect to each other. Interference has been the main limiting factor ever since wireless communication evolved and its characterization presents critical challenges \cite{gibson2012mobile}. Characterizing interference field accurately is an important step toward the design and deployment of wireless networks. It is a challenging problem which has led to an increasing interest in modeling interference in wireless networks. Network topology which determines the spatial distribution of transmitter-receiver pairs is highly affected by the network characteristics such as network infrastructure.

The received signal-to-interference-plus-noise ratio ($\sinr$) is known to be a strong indicator of the performance and reliability of a wireless link. Several key performance metrics of interest, such as outage probability, ergodic capacity, and outage capacity, are strongly dictated by the received $\sinr$. By definition, the $\sinr$ distribution depends upon the joint distribution of the received powers from the serving node and the interfering nodes, which ultimately depend on the network topology. Therefore, accurate modeling of the network topology becomes a key step towards meaningful performance analysis of wireless networks. Owing to its tractability and realism in modeling irregular node locations, stochastic geometry has emerged as an important tool for the realistic analysis of wireless networks~\cite{chiu2013stochastic,haenggi2012stochastic,baccelli2009stochastic,win2009mathematical}. Initially popular for the modeling of wireless ad hoc and sensor networks, e.g. see~\cite{haenggi2009interference,andrews2012transmission}, it has recently been adopted for the analysis of cellular and heterogeneous cellular networks as well~\cite{andrews2011tractable,dhillon2012modeling,mukherjee2012sinr}. Irrespective of the nature of the wireless network, the interference field is almost always modeled by a homogeneous PPP to maintain tractability. While this leads to remarkably simple results for key performance metrics, such as coverage and rate, it is not quite suitable for modeling the effect of interference management techniques, which often introduce some form of spatial interaction among transmitters. In this paper, we focus on {\em spatial separation}, where holes (also called {\em exclusion zones}) are created around nodes/networks that need to be protected from excessive interference \cite{vu2008primary}. In particular, we assume that the {\em baseline} interference field is a PPP from which holes of a given radius are carved out. When the locations of the holes also form an independent PPP, the resulting point process is usually termed as a PHP, which is the main focus of this paper. 

\subsection{Related Work and Applications}
We first discuss a few of possibly numerous instances in wireless networks where PHP is a more appropriate model for node locations. In particular, we discuss how PHP has been used to model cognitive radio networks, heterogeneous cellular networks, and device-to-device (D2D) networks. The main objective of cognitive radio networks is to improve spectrum utilization by allowing unlicensed {\em secondary} users to use licensed spectrum as long as they do not cause excessive interference to the licensed {\em primary} users. One of the ways to ensure this is by creating exclusion zones (holes) around primary users, where secondary transmissions are not allowed. This spatial separation was elegantly modeled by using a PHP in \cite{lee2012interference}. In particular, the locations of both primary and secondary users were first modeled by independent PPPs. Assuming secondary transmissions were not allowed within a given distance from the primary users, the locations of {\em active} secondary users were then modeled by a PHP. 
%In cognitive radio networks, holes (exclusion zones) are established around primary links, which confine active secondary transmissions outside the holes. Modeling the locations of the primary nodes by a PPP, the locations of the {\em active} secondary nodes then follow a PHP (assuming the {\em base process} of secondary users was a independent PPP). This model was first proposed for cognitive radio networks in \cite{lee2012interference}. 

The PHP has also been used recently to model inter-tier dependence in the base station locations in a heterogeneous cellular network in~\cite{deng2014heterogeneous,dengheterogeneousax,Hybrid_ICC}. Modeling the macrocell locations by a PPP, it was assumed that the small cells are deployed farther than a minimum distance from the macrocells, i.e., outside exclusion zones of a given radius. In such a case, small cells form a PHP. This model introduces repulsion between the locations of macro and small cells, which is desirable due to several reasons, such as interference mitigation at macrocells due to small cell transmissions, and the higher advantage of deploying small cells towards the cell edges of macrocells, especially in the coverage-centric deployments. 
%In addition, PHP is used in a two tier Heterogeneous network comprising macro-cells and small cells that operate under the spectrum underlay strategy~\cite{ deng2014heterogeneous,dengheterogeneousax,Hybrid_ICC}.  Here, interference mitigation schemes in a two-tier heterogeneous network establish holes around the active nodes of the macro-tier whereas the remaining small cell transmitters form a PHP.
%As a result of applying interference mitigation techniques in a two-tier Heterogeneous network, holes are formed around active nodes of the macro-tier whereas the remaining small cell transmitters form PHP.
%%%%%%
%Further, for underlay device-to-device (D2D) communications in cellular networks, inhibition zones around cellular links may be established to save cellular communications from D2D interference when D2D transmissions are not allowed inside the zones. The active D2D transmitters outside the zones form PHP~\cite{sun2014d2d,sakr2014cognitive,elsawy2014analytical}.  In~\cite{AfsDhiC2015a},   a Poisson Cluster Process (PCP) meets Poisson Hole Process to develop a new spatial model for integrated D2D and cellular networks. Here, the location of devices is captured by a modified Thomas cluster process models while the cluster centers are modeled by a PHP.  The holes model inhibition zones wherein D2D transmissions are not allowed to protect cellular links.

Similarly, for underlay D2D
communication in cellular networks, inhibition zones may be created around cellular links where no D2D transmissions are allowed, thus saving cellular links from excessive D2D interference.  
The active D2D transmitters outside the holes form a PHP~\cite{sun2014d2d, elsawy2014analytical}. 
In this regard, cognitive D2D communication in cellular network when transmitters are powered by harvesting energy from the ambient interference is studied in~\cite{sakr2014cognitive}.
%Here, D2D transmitters which are located at a minimum distance from cellular base stations  form a PHP.
In~\cite{AfsDhiC2015a}, a Poisson Cluster Process (PCP) and a PHP are merged to develop a new
spatial model for integrated D2D and cellular networks. 
In particular, a modified Thomas cluster process is used to model device locations where instead of modeling the cluster centers as a homogeneous PPP, they are modeled as a PHP to account for the inhibition zones around cellular links. 

Despite the importance of a PHP in modeling wireless networks, the exact characterization of interference experienced by a typical receiver in a PHP is a challenging problem. There are two main directions taken in the literature for the analysis of wireless networks modeled by PHPs. The first approach, termed {\em first-order statistic approximation}, approximates PHP by a homogeneous PPP with the same density~\cite{haenggi2012stochastic}. The second approach ignores the holes altogether to approximate the PHP by its {\em baseline} PPP. This overestimates the interference and the accuracy of the approximation is a function of system parameters \cite{lee2012interference,sun2014d2d,elsawy2014analytical,sakr2014cognitive}. Besides, the PHP is sometimes approximated with a PCP by matching the first and second order statistics~\cite{lee2012interference,deng2014heterogeneous,dengheterogeneousax}. The resulting expressions for performance metrics are usually more complicated in this case compared to the above two. While all these approaches are reasonable, they are typically not accurate beyond a specific range of system parameters. In this paper, we take a fresh look at this problem and derive tight upper and lower bounds for the Laplace transform of interference experienced by the typical node in a PHP. The main contributions are summarized next.

%As a result, all the state-of-the-art approaches involve some form of approximations to deal with a PHP. 
%For instance, one typical direction is to consider a first-order statistic approximation based on independent thinning of a PPP~\cite{haenggi2012stochastic}. 
%Second direction is to ignore the holes and use the resulting PPP to bound the interference experienced by a typical receiver  ~\cite{lee2012interference,sun2014d2d,elsawy2014analytical}. 
%In addition, the PHP is sometimes approximated with a Poisson Cluster Process by using moment-matching based approaches which leads to more complicated expressions for the performance metrics compared to the above approaches~\cite{lee2012interference,deng2014heterogeneous,dengheterogeneousax}.
%While all these approaches are reasonable, they are typically not accurate beyond a specific range of system parameters. 

\subsection{Contributions}

\noindent {\em New approach to the analysis of a PHP.} 
Unlike existing approaches that approximate a PHP with either a PPP or a PCP, we develop a new approach that is amenable to shot-noise analysis and leads to tight provable bounds on the Laplace transform of interference experienced by a typical node of a PHP. A lower bound is first derived by overestimating interference by ignoring all the holes except the closest one. We provide an equivalent interpretation of this result in which the closest hole is {\em dissolved} in such a way that it results in a tractable non-homogeneous PPP. The resulting bound is shown to be remarkably tight. Extending this approach to multiple holes, we derive an upper bound on the Laplace transform of interference by carving out each hole separately without accounting for the overlaps between them. This leads to the removal of some points from the baseline PPP multiple times, thus underestimating the interference power experienced by the typical point. This bound is also shown to be remarkably tight across a variety of scenarios, including the ones in which the holes exhibit significant overlaps.

\noindent {\em Approaches to incorporate the effect of overlaps in the holes.} In the first set of bounds discussed above, we carefully circumvented the need for incorporating the effect of overlaps between holes. While these simple and easy-to-use bounds are tight, we also provide ideas for incorporating the overlaps between holes, which tighten these bounds even further. In the first result, we generalize the lower bound discussed above by considering two closest holes in the interference field while incorporating the {\em exact} effect of overlap between them. In the second result, instead of trying to incorporate the exact effect of overlaps, we propose a new procedure for bounding the overlap area, which allows us to derive a provable lower bound on the Laplace transform while considering multiple holes in the interference field. In the third and final result, we propose a new approach that allows to incorporate the {\em mean} effect of overlaps in the holes.

\noindent {\em New insights.} Our results concretely demonstrate that for accurate analysis of interference in a PHP, it is very important to preserve the local neighborhood around the typical point. For instance, we show that considering even a single hole in the interference field results in a tighter characterization of interference power at the typical point of a PHP compared to seemingly more refined prior approach of first-order statistic approximation in which the PHP is approximated by a PPP with the matching density. This is because by considering a single hole, the local neighborhood around the typical point is accurately captured, whereas it is distorted in the other approach due to independent thinning involved in the density matching of a PPP. 
%%%% Include the following text at the time of submission. 
%Numerical results also reveal that our first set of bounds derived without incorporating the effect of overlaps, while being seemingly simple, are so tight that the additional complication in the expressions resulting from more sophisticated analyses of overlaps is not commensurate with the gains.

\begin{table}
\caption{Notation and Network Parameters}
\begin{center}
\scalebox{.81}{%
\begin{tabular}{l|l}
 \hline
  \hline
Symbol & Description \\
 \hline 
 $\Phi_1; \lambda_1$ & Independent PPP modeling the locations of hole centers; density of $\Phi_1$ \\
  $\Phi_2; \lambda_2$ & Independent PPP from which the holes
are carved out; density of $\Phi_2$ \\
$\Psi$ & PHP formed by carving out holes with centers $\Phi_1$ from  $\Phi_2$\\
$D$ & Radius of each hole carved out from  $\Phi_2$\\
$\pc, \gamma$ & Coverage probability (in terms of SIR); SIR threshold \\
  $\ncalL_{{I}}(s)$ & Laplace transform of  $I$, defined as $\E\left[e^{sI}\right]$  \\
  $\ncalC=\nbb(\nby,D)$ & Ball of radius $D$ centered at $\nby$ \\
%  $\E, f_X(.)$ & Expectation; pdf of the random variable X \\
%  $\|.\|$ & Euclidean distance \\
  $h_\nbx$ & Fading gain; $h_\nbx \sim \exp(1)$ for Rayleigh fading\\
  $\alpha$ & Path-loss exponent for all the wireless links \\
  $P; r_\mathrm{0}$ &Transmit power; serving distance for the link of interest \\
   \hline
 \hline
 \end{tabular}
 }
\end{center}
\label{Tab: Network Parameters2}
\end{table}
\section{Network Model}
\label{NetModel}
\subsection{System Model}
\label{Systmmodel}
%\begin{defn} [COX Process] Let $Q$ be a distribution  on $\mathcal{S}\subset\mathbb{R}^d$ which denotes a space of non-negative random measures. 
%Given the distribution of the Poisson point process $P_\Lambda$ of density measure $\Lambda$, and $\Psi$ a random measure with distribution \textit{Q} (7.1.2 in ~\cite{chiu2013stochastic}), then the Cox process $\Phi$ with driving random measure \textit{Q} is defined as follows
%\begin{align}
%P_{\Phi}(Y)=\int_{\mathcal{S}}P_{\Lambda}Q(\mathbb{d}{\Lambda})
%\end{align}
%The COX process is derived in a two step mechanism: 
%first, a random measure $\Lambda$ is generated based on the random measure distribution \textit{Q} and then, the Poisson process of density measure $\Lambda$ is generated. 
%Hence, Cox process is often called as the doubly stochastic Poisson process.
%\end{defn}
We consider a wireless network that is modeled by a PHP in $\mathbb{R}^2$. A PHP can be formally defined in terms of two independent homogeneous PPPs $\Phi_1$ and $\Phi_2$, where $\Phi_2$ represents the {\em baseline} PPP from which the holes are carved out and $\Phi_1$ represents the locations of the holes. Let the densities of $\Phi_1$ and $\Phi_2$ be $\lambda_1$ and $\lambda_2$, respectively, with $\lambda_2>\lambda_1$. Denoting the radius of each hole by $D$, the region covered by the holes can be expressed as
\begin{align}
\Xi_D \triangleq  \bigcup_{\nby \in \Phi_1} \mathbf{b}(\nby,D), \quad \mathbf{b}(\nby,D) \equiv \{\nbz \in \nbbR^2: \|\nbz-\nby\| < D\}.
\label{eq:XiD}
\end{align}
The points of $\Phi_2$ lying outside $\Xi_D$, form a PHP, which can be formally expressed as
\begin{align}\label{Psidef}
\Psi =\{\nbx\in \Phi_2 : \nbx\notin \Xi_D\}= \Phi_2 \setminus \Xi_D.
\end{align}
It should be noted that the PHP $\Psi$ has also been known as a {\em Hole-1 process} in the literature~\cite{ganti2006regularity}. 

We characterize the interference experienced by a {\em typical node} in $\Psi$ due to the transmission of the other nodes of $\Psi$. Due to the stationarity of the process, the typical node can be assumed to lie at the origin $\mathbf{o}$, and due to Slivnyak's theorem, we can condition on $\mathbf{o} \in \Psi$ without changing the distribution of the rest of the process~\cite{haenggi2012stochastic}. This equivalently means that the interferers are modeled by the PHP $\Psi$ with the typical receiver being an additional point placed at the origin.
%Due to Slivnyak's theorem for a PPP~\cite{haenggi2012stochastic}, we can condition on the typical point without changing the distribution of the rest of the process. This typical point can further be placed at the origin  due to the stationarity of the process. This equivalently means that the interferers can be modeled by the PHP $\Psi$ with the typical receiver being an additional point placed at the origin. 
Note that Slivnyak's theorem for a PPP is applicable here because conditioned on the locations of the holes (i.e., conditioned on $\Phi_1$), PHP $\Psi$ is simply a PPP of density $\lambda_2$ defined on $\nbbR^2 \setminus \Xi_D$. Since the typical point is located outside the holes by construction, there are no points of $\Phi_1$ within a disk of radius $D$ around the typical point. For this receiver, we assume that the serving transmitter is located at the fixed distance $r_0$.  It should be noted that we could have considered more sophisticated models for the serving link of interest but we chose to consider this simple setup because our emphasis is on characterizing interference in a PHP.

%%% DO NOT DELETE THIS TEXT. MAY BE USEFUL LATER %%%
%%%%% Useful text for connection with the bipolar model %%%%%
% Besides, note that this setup has some connections with the well-known Poisson bipolar model, where each transmitter has a receiver at a fixed distance $r_0$ in an orientation that is chosen uniformly at random with respect to the corresponding transmitter~\cite{haenggi2012stochastic}. The reason we do not claim our setup to be studying the performance of a typical link in a Poisson bipolar analogue of a PHP is because we focus on the receiver that is located outside the holes, which wouldn't have been the case in the bipolar model. Nevertheless, readers conversant with stochastic geometry-based analyses of wireless networks will note that our setup can be trivially extended to study the performance of this typical link in a Poisson bipolar analogue of a PHP.
%%%%%%%%%%%%%%%%%%%%%%%%%%%%%%%%%%%

For the wireless channel between points $\nbx$ and $\nby$, we consider a standard power law path-loss $l(\nbx-\nby)=\|\nbx-\nby\|^{-\alpha}$ with path-loss exponent $\alpha$. All the wireless links are assumed to experience independent Rayleigh fading. All the transmitters are assumed to transmit at a fixed power $P$. The received power at the typical node from its transmitter of interest is therefore $P_r = Phr_0^{-\alpha}$, where $h \sim \exp(1)$ models
Rayleigh fading. Similarly, the interference power experienced by the typical receiver located at the origin is 
%To calculate the interference power at the receiver, we assume that the receiver of interest is located at the origin while it is located outside the holes.
%We consider a reference transmitter which serves the receiver of interest at $\nbx_0$ (where $\|\nbx_0\| = r_0$). 
%Hence, interference is received from all of the transmitters except the serving one. 
%For the transmitter located at $\nbx\in \Psi$,  denote the fading coefficient on the link to the receiver of interest by $h_\nbx$, and assume that the path-loss exponent from all transmitters in $\Psi$ is $\alpha$.  Then, the total interference at the receiver of link of interest from all interferers is
\begin{equation}\label{Ij1}
    I=\sum_{\substack{\nbx\in {\Psi}}}P h_\nbx \|\nbx\|^{-\alpha},
 \end{equation}
 where $h_\nbx \sim \exp(1)$ models Rayleigh fading gain for the link from interferer $\nbx\in {\Psi}$ to the typical receiver. For this setup, we define coverage probability next.

\subsection{$\sir$ and Coverage Probability} 
Using the received power over the link of interest and the interference power defined in the previous subsection, the \emph{signal to interference ratio} ($\sir$) can be expressed as
 \begin{align}\label{SIR}
    \mathtt{SIR}(r_0) = \frac{P h {r_0}^{-\alpha}}{ \sum_{\substack{\nbx\in {\Psi}}}P h_\nbx \|\nbx\|^{-\alpha}}.
 \end{align}
Denote the minimum $\sir$ required for successful decoding and demodulation at the typical receiver by $\gamma$. A useful metric of interest in wireless networks is the \emph{$\sir$ coverage probability} $\pc$, which is the probability that the $\sir$ at the receiver exceeds the threshold $\gamma$. Mathematically, 
\begin{align}\nonumber 
\pc&=\,\mathbb{P}\{\mathtt{SIR}(r_0)>\gamma\,\} =\mathbb{P}\bigg\{h >\frac{\gamma r_0^{\alpha}}{P}I\bigg\}\\   \label{eq:Pcdef}
&\stackrel{(a)}{=}\E\left[ \exp\left(-\frac{\gamma r_0^{\alpha}}{P} I\right)\right] \stackrel{(b)}{=} \mathcal{L}_{I}\left(\frac{\gamma r_0^{\alpha}}{P}\right) , 
\end{align}
where $(a)$ follows from the fact that $h\sim \exp(1)$, and $(b)$ from the definition of Laplace transform of interference power $\mathcal{L}_{I}(s) = \mathbb{E}[\exp(-s I)]$. Note that for this setup, it is sufficient to focus on the Laplace transform of interference in order to study coverage probability. In general, accurate characterization of $\mathcal{L}_{I} (s)$ is the first step in the analysis of more general classes of wireless networks, including cellular networks~\cite{dhillon2012modeling}. Therefore, we will focus on $\mathcal{L}_{I} (s)$ in the technical sections of the paper with the understanding that the coverage probability can be easily derived for our setup using \eqref{eq:Pcdef}. We begin our technical discussion by summarizing two key prior approaches used in the literature for characterizing $\mathcal{L}_{I} (s)$ in a PHP. 

For the ease of reference, the notation used in the paper is summarized in Table~\ref{Tab: Network Parameters2}.

%We begin by summarizing two popular prior approaches used to characterize $\mathcal{L}_{I} (s)$ for PHPs in the literature.

%\section{Prior Approaches to the Laplace Transform of Interference in a PHP}
\section{Key Prior Approaches}
\label{Approx_Lower}
%\subsection{Network Model}
In this section, we summarize two popular approaches that have been used in the literature to derive the Laplace transform of interference in a PHP. At the end of the Section, we also provide insights into the strengths and weaknesses of each approach. 
%The performance analysis of PHP model is performed by using bounding and approximation techniques.
%This is because the unknown probability generating functional of PHP. 
%In this section, we first investigate the existing approaches for the modeling and analysis of PHP to achieve a fair comparison between the prior and proposed approaches. 
%We state the strengths and weaknesses of each approach. 
%\subsection{Laplace Transform of Interference in a PPP}
\subsection{Lower Bound by Ignoring Holes: Approximating $\Psi$ by $\Phi_2$}
The first approach is to ignore the effect of holes and approximate the interference field $\Psi$ by the baseline PPP $\Phi_2$ of density $\lambda_2$. By construction, this approach overestimates the interference power and hence leads to the lower bound on the Laplace transform of interference \cite{lee2012interference}. This well-known result is stated below for completeness.
%The first approach for the analysis of the Laplace transform of PHP is ignoring the impact of holes. As a direct consequence, PHP leads to a PPP with density of $\lambda_2$.  Using the PPP model, the lower bound for the Laplace transform of interference corresponding to PHP is given as follows \cite{lee2012interference}. It is due to the fact that the density of the transmitters in a PHP model is overestimated by the PPP model.
\begin{lem} [{Lower bound}] \label{thm :lowbound}
Ignoring the impact of holes (approximating $\Psi$ by $\Phi_2$), the Laplace transform of aggregate interference $I=\sum_{\substack{\nbx\in {\Psi}}}P h_\nbx \|\nbx\|^{-\alpha}$ is lower bounded   by:
\begin{align}\label{eq: LTfP00}
\ncalL_{I}(s)\geq \exp\left[-\pi\lambda_2 \frac{(sP)^{2/\alpha}}{\mathrm{sinc}(2/\alpha)}\right].
\end{align}
%where $s=\frac{\gamma r_\mathrm{0}^{\alpha}}{P}$.
\end{lem}
\begin{proof}  See Appendix~\ref{app: A}.
\end{proof}
The tightness of the above bound will be demonstrated in the Numerical Results section.%
%\begin{remark}
%The PPP model gives a lower bound for the coverage probability. 
%%This is due to the fact that in this model, the impact of exclusion zones is ignored. 
%This approach leads to the overestimation of the interference in PHP.
%\end{remark}
\subsection{Approximating PHP by a PPP with the Same Density}
\label{LTI FOSA} 
The second approach to the derivation of the Laplace transform of PHP is the {\em first-order statistic approximation} \cite{haenggi2012stochastic}. In this approach, the baseline PPP $\Phi_2$ is independently thinned such that the resulting density of the PPP is the same as that of the PHP $\Psi$, which we denote by $\lambda_{\mathrm{PHP}}$. The first step in this approach is therefore to derive $\lambda_{\mathrm{PHP}}$ in terms of the given system parameters, which was done in \cite{haenggi2012stochastic}. For completeness, we discuss its proof briefly below. To derive $\lambda_{\rm PHP}$, we first need to derive an expression for the average number of points of the PHP $\Psi$ lying in a given set $\ncalB \subset \nbbR^2$, which by definition is
\begin{align}
&\nonumber
\E \left[ \sum _{\nbx\in \Phi_2\cap \ncalB}\prod _{\nby\in\Phi_1} (1-\nb1_{\mathbf{b}(\nbx,D)}(\nby))\right] \\ \nonumber
& \stackrel{\text{(a)}}= \E_{\Phi_2} \left[\sum _{\nbx\in \Phi_2\cap \ncalB} \E_{\Phi_1} \left[\prod _{\nby\in\Phi_1} (1-\nb1_{\mathbf{b}(\nbx,D)}(\nby))\right]\right]\\\nonumber
&\stackrel{\text{(b)}}=\E_{\Phi_2}\left[ \sum_{\nbx\in\Phi_2\cap \ncalB} \exp\left(-\lambda_1\int_{\R^2} \nb1_{\mathbf{b}(\nbx,D)}(\nby) \nrmd \nby\right)\right]
\\ \nonumber
&
\stackrel{\text{(c)}}= |\ncalB| \lambda_2 \exp(-\lambda_1\pi D^2),
\end{align}
where (a) is due to the independence of point processes $\Phi_1$ and $\Phi_2$, 
(b) follows from the probability generating functional (PGFL) of a PPP, and (c) follows from the Campbell theorem~\cite{chiu2013stochastic}. From the above expression, we can readily infer that $\lambda_{\rm PHP} = \lambda_2 \exp(-\lambda_1\pi D^2)$. Now to derive the Laplace transform of interference in this case, we just need to replace $\lambda_2$ in the result of Lemma~\ref{thm :lowbound} with $\lambda_{\rm PHP}$. The result is stated below for completeness. 
%It {\em seemingly} refines the above approach by approximating the PHP by a PPP
%The second approach {\em seemingly} refines the above approach by approximating PHP by a PPP with the same density of the points. 

\begin{lem} [{Approximation}] \label{thm :appr}
The Laplace transform of interference power $I=\sum_{\substack{\nbx \in {\Psi}}}P h_\nbx \|\nbx\|^{-\alpha}$ when PHP $\Psi$ is approximated by a PPP with density $\lambda_{\rm PHP}$ is
%The Laplace transform of aggregate interference corresponding to a PHP model, $I=\sum_{\substack{\nbx \in {\Psi}}}P h_\nbx \|\nbx\|^{-\alpha}$ can be approximated as:
\begin{align}\label{eq: LTfP02}
\ncalL_{I}(s )\simeq\exp\left[-\pi\lambda_{\mathrm{PHP}} \frac{(sP)^{2/\alpha}}{\mathrm{sinc}(2/\alpha)}\right].
\end{align}
%\begin{proof}  This approximation is derived by substituting $\lambda_{\mathrm{PHP}} \rightarrow \lambda_{\mathrm{2}}$ in Lemma \ref{thm :lowbound}. 
%\end{proof} 
\end{lem}
\begin{remark} \label{rem:1}
Both the approaches discussed above approximate $\Psi$ with a homogeneous PPP: the one first with density $\lambda_2$ (the baseline PPP), and the second one with density $\lambda_{\rm PHP} < \lambda_2$ (the density of the PHP $\Psi$). While the second approach is a seemingly more refined approach, a careful thought reveals that in order to match the density of the PPP with that of a PHP, the baseline PPP has to be independently thinned, which disturbs the local neighborhood of points around the typical point, thus resulting in a loose bound. On the other hand, approximating $\Psi$ simply by the baseline PPP $\Phi_2$ preserves the local neighborhood resulting in a relatively tighter approximation. More insights will be provided in the numerical results section. 
\end{remark}

Before concluding this section, it is important to note that there is one more fitting-based approach in which the PHP is approximated by a PCP with matching first and second order statistics. The Laplace transform of interference and other performance metrics are then studied using the fitted PCP. Since the formal description of this technique requires significantly more details compared to the techniques discussed above, we refer the reader to~\cite{lee2012interference,deng2014heterogeneous,dengheterogeneousax}, where this approach has been used for the analysis of cognitive radio and heterogeneous cellular networks modeled as PHPs. We now discuss the proposed approaches next.

\section{Proposed Approaches to Laplace Transform of Interference in a PHP}

\begin{figure}[t!]
  \centering{
 \includegraphics[width=.65\linewidth]{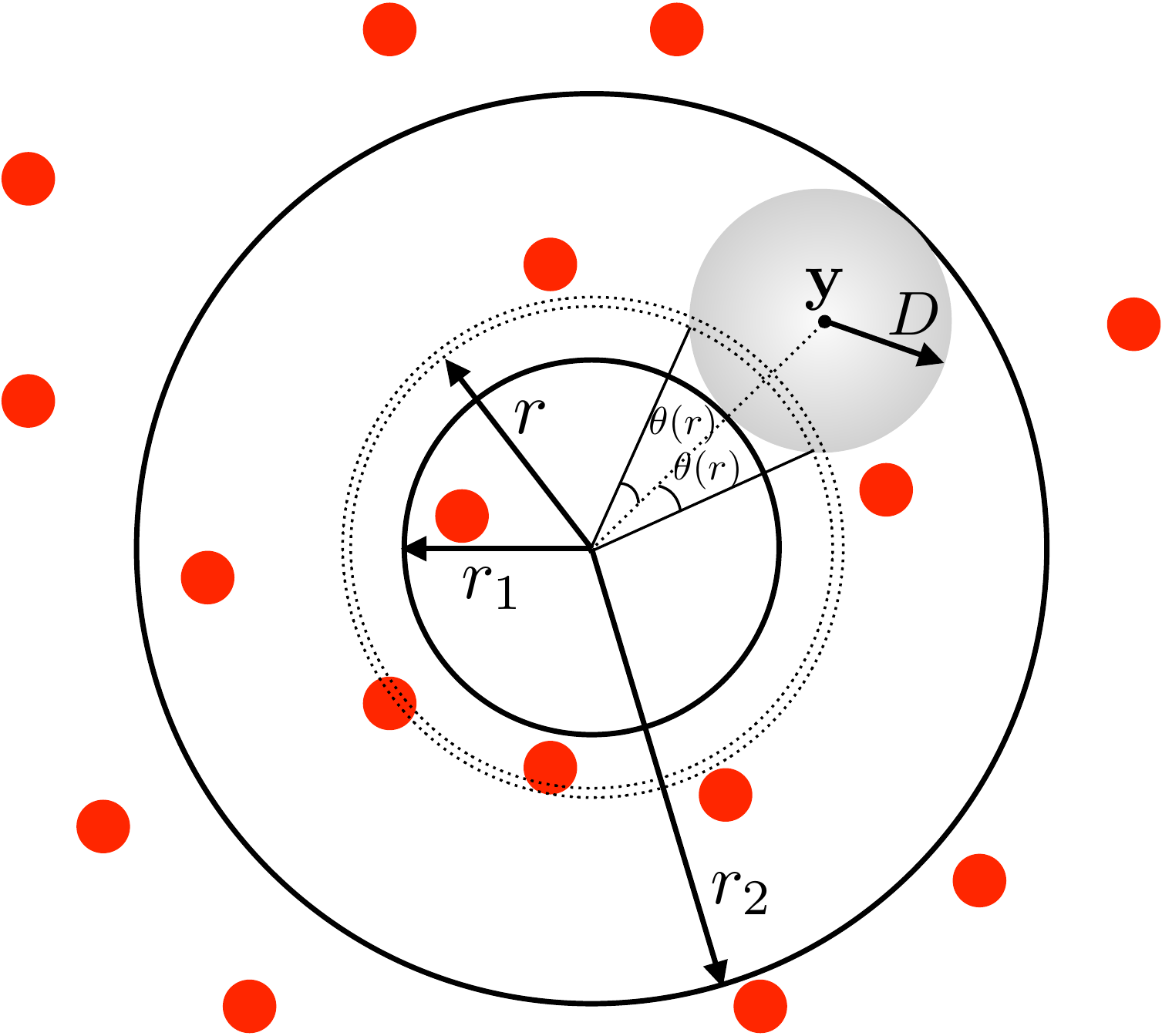}
              \caption{Illustration of the interference field with a single hole.}
\label{proof04}}
\end{figure}

\label{PerfAnl} 
%\subsection{New Modeling Approach}
We now introduce our proposed approach to characterize the Laplace transform of interference in a PHP. In the first intermediate step,  we model the locations of interferers by a homogeneous PPP $\Phi_\mathrm{2}$ of density $\lambda_\mathrm{2}$ from which only one hole $\ncalC$ of radius $D$ is carved out at a {\em deterministic} location. Let the location of the center of this hole be $\nby \in \nbbR$ and hence its distance from the origin be $\|\nby\|$. The resulting setup is illustrated in Fig.~\ref{proof04}.
Note that the interference field in this case is non-isotropic due to the fixed location of the hole. The Laplace transform of the interference power at the origin from the nodes of $\Phi_\mathrm{2}$ outside $\ncalC$ is characterized next. This intermediate result will be used later in this section to derive upper and lower bounds on the Laplace transform of interference experienced by a typical node in a PHP.
%In the sequel, we first derive an expression for the distribution of the aggregate interference after removing the closest hole to the typical point from $\Phi_\mathrm{2}$.  Then, we provide our proposed approach for the performance analysis of the system involving a hole.
\begin{lem} \label{Lemma1_Proof}
Let $I=\sum_{\substack{\nbx\in {\Phi_\mathrm{2}\cap\nbb^c(\nby,D)}}}Ph_\nbx\|\nbx\|^{-\alpha}$, the Laplace transform of interference conditioned on $\|\nby\|$ is
\begin{align}  \nonumber  
\ncalL_{{I|\|\nby\|}}(s)=&\exp\left(-\pi\lambda_2 \frac{{( sP)}^{2/\alpha}}{\sinc(2/\alpha)}\right)\times
\\ \label{LtI asym_03}
&\exp\left(\int_{\|\nby\|-D}^{\|\nby\|+D}\frac{2\pi\lambda(r)}{1+\frac{r^\alpha}{Ps}} r\nrmd r\right)
\end{align}
where $\lambda(r)= \frac{\lambda_2}{\pi}{\arccos\left(\frac{r^2+\|\nby\|^2-D^2}{2\|\nby\|r}\right)}$, 
$\ncalC={\mathbf{b}(\nby,D)}$ denotes the hole centered at $\nby$ with radius $D$.
\end{lem}
\begin{proof}  See Appendix~\ref{app: B}.
\end{proof}
%%%%%%%%%%%%%%%%%%%%%%%%%%%%%%%%%%%%%%%%%%%%%
\begin{remark}[Dissolving the hole]
The above result has an interesting interpretation that will be useful in visualizing the proposed results. Note that since received power is a radially symmetric function, it solely depends upon the distance of the transmitter to the origin. Therefore, we can in principle, dissolve the hole as long as the number of points lying in a thin strip of radius $\|\nby\|-D\leq r\leq \|\nby\|+D$ and vanishingly small width ${\rm d} r$ is not changed. Please refer to Fig.~\ref{proof04} for an illustration of this strip. Taking a closer look at the interference originating from this strip, we note that the only thing that matters is the number of points distributed in the part of the strip which is outside the hole. The area of this region is $2r {\rm d} r (\pi - \theta(r))$, where the angle $\theta(r) = \arccos\left(\frac{r^2+\|\nby\|^2-D^2}{2\|\nby\|r}\right)$ is defined in Fig.~\ref{proof04}. Therefore, the number of interfering points lying within this strip is Poisson distributed with mean $\lambda_2 2r {\rm d} r (\pi - \theta(r))$. Since the exact locations of these points within the strip doesn't matter, we can dissolve the hole and redistribute the points uniformly inside the whole strip of area $2\pi r {\rm d}r$. This means, the PPP with a hole can be equivalently modeled as a non-homogeneous PPP with density $\lambda_2(1-\theta(r)/\pi)$, where the $\lambda_2\theta(r)/\pi$ term (defined as $\lambda(r)$ in Lemma~\ref{Lemma1_Proof}) captures the effect of hole. 
\end{remark}

Using this intermediate result and the above insights, we now derive tight bounds on the Laplace transform of interference experienced by a typical node in a PHP.
%%%%%%%%%%%%%%%%%%%%%%%%%%%%%%%%%%%%%%%%%% 
\subsection{Lower Bound on the Laplace Transform of Interference in a PHP}
\label{Sec: FIRST-Closest}
Before going into the technical details, note that due to path-loss, the effect of holes that are close to the typical point will be much more significant compared to the holes that are farther away. Therefore, to derive an easy-to-use lower bound on the Laplace transform of interference, we consider only one hole; the one that is closest to the typical point; and ignore the other holes. Denoting the location of the closest hole by $\nby_1$, the interference field in this case is $(\Phi_2 \cap \nbb^c(\nby_1,D)) \supset \Psi$, which clearly overestimates the interference of PHP and hence leads to a lower bound on the Laplace transform. Note that in Lemma~\ref{Lemma1_Proof}, we have already derived the {\em conditional} Laplace transform for the case when there is one hole and its distance to the origin is known. To derive a lower bound, we simply need to assume this hole to be the closest point of $\Phi_1$ to the origin and decondition the result of Lemma~\ref{Lemma1_Proof} with respect to the distribution of $V_1 = \|\nby_1\|$. 
To this end, we first derive the probability density function (PDF) of $V_1$ next.
\begin{lem} \label{Clos_Dislem}
The PDF of the distance $V_1=\|\nby_1\|$ between the typical node at the origin and the closest point of $\Phi_1$ is given by
\begin{align}	\label{Eq: FRcrofton}
f_{V_1}(v_1)=2\pi\lambda_1 v_1\exp(-\pi\lambda_1 (v_1^2-D^2)), \quad v_1 \geq D.
\end{align}
\begin{proof}
 \begin{figure}[t!]
  \centering{
              \includegraphics[width=.65\linewidth]{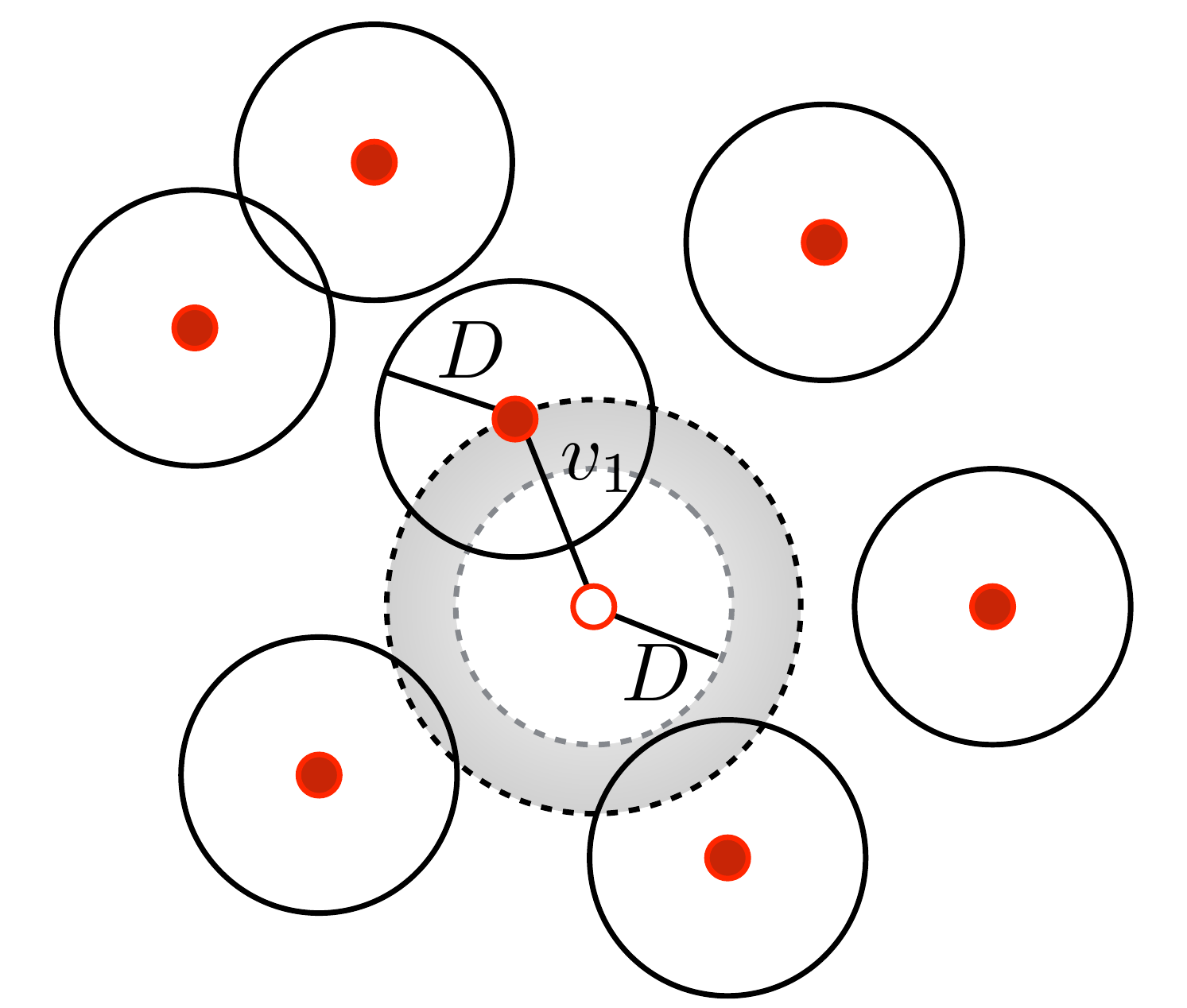}
              \caption{Illustration showing that the closest point of $\Phi_1$ is at least a distance $D$ away from the typical point of the PHP $\Psi$.}
\label{first-closest-dis7}}
          \end{figure}   

As discussed in Section~\ref{NetModel}, the typical point of a PHP lies outside the holes by construction. Therefore, as illustrated in  Fig.~\ref{first-closest-dis7}, the minimum distance between the typical point and the closest hole (closest point of $\Phi_1$) is $D$. Using this fact along with the properties of a PPP, the distribution of $V_1$ can be derived as follows: 
% \label{Eq: clproof5} 
\begin{align*}
\nbbP(V_1 > v_1) &= \nbbP(\text {Number of points of $\Phi_1$ in the set}\\\nonumber
&\ \{\nbb(0,v_1)\setminus \nbb(0,D)\} = 0)\\
&= \exp(-\pi\lambda_1 (v_1^2-D^2)), \quad v_1 \geq D.
\end{align*}         
The result now follows by differentiating the above expression.   
\end{proof}
\end{lem}
Deconditioning the result of Lemma~\ref{Lemma1_Proof} with respect to the distribution of the distance to the closest hole derived above, the proposed lower bound is derived below.
\begin{theorem} [{New Lower Bound~1}]\label{thm :nearhole}
Let $I=\sum_{\substack{\nbx\in {\psi}}}Ph_\nbx\|\nbx\|^{-\alpha}$, the Laplace transform of interference is lower bounded by
\begin{align}\nonumber
&\ncalL_{{I}}(s)\geq \exp\left(-\pi\lambda_2 \frac{{( sP)}^{2/\alpha}}{\sinc(2/\alpha)}\right)\times \\  \label{LtI asym_3}
& \int_D^{\infty}\exp\left(g(v_1)\right)2\pi \lambda_1 v_1\exp(-\pi\lambda_1 (v_1^2-D^2))\nrmd v_1,
\end{align}
where $g(v_1)=\int_{{v_1-D}}^{{v_1+D}}{\arccos\left(\frac{r^2+v_1^2-D^2}{2v_1r}\right)}\frac{2\lambda_2}{1+\frac{r^\alpha}{Ps}} r\mathrm{d}r$.
\end{theorem}
\begin{proof} See Appendix~\ref{app: C}.
\end{proof} 
\begin{remark}\label{remark3_r3}
Since the PHP is approximated by $\Phi_2 \cap \nbb^c(\nby_1,D)$ in the above result, the resulting lower bound presented in Theorem~\ref{thm :nearhole} is by construction tighter than the known lower bound presented in Lemma~\ref{thm :lowbound} where the interference field was approximated by simply $\Phi_2$. On the same lines as discussed for Lemma~\ref{thm :lowbound} in Remark~\ref{rem:1}, the above approach captures the local neighborhood of the typical node accurately, thus leading to a remarkably tight lower bound in  Theorem~\ref{thm :nearhole}. This will be demonstrated later in this section and in the numerical results section. 
%The new lower bound derived in Theorem~\ref{thm :nearhole} above is by construction tighter than the known lower bound presented in Lemma~\ref{thm :lowbound}. Note that due to pathloss, the nodes located in the proximity of the typical node are more important for interference characterization. Since the new bound captures the local neighborhood around the typical node accurately, this is expected to be a tight bound. This will be numerically validated for a variety of cases in the numerical results section. 
\end{remark}
\subsection{Upper Bound for the Laplace Transform of Interference in a PHP}
To derive an upper bound on the Laplace transform, we extend the above approach to all the holes. To maintain tractability, each hole is carved out individually/separately from the baseline PPP $\Phi_2$ using the above approach. Note that since the centers of the holes follow a PPP $\Phi_1$, there will obviously be overlaps among holes. Therefore, when we remove points of $\Phi_2$ corresponding to each hole individually (without accounting for the overlaps), we may remove certain points multiple times thus underestimating the interference field, which results in an upper bound on the Laplace transform of interference.
In the next section, we show that any reasonable attempt towards incorporating the exact effect of overlaps leads to a significant loss in tractability. 
Fortunately, the bounds derived in this section without incorporating the effect of overlaps are remarkably tight and can be considered proxies for the exact Laplace transform. 
\begin{theorem} [{New Upper Bound}] \label{IccThm3} The Laplace transform of interference experienced by a typical node in a PHP is upper bounded by
\begin{align}\nonumber
\ncalL_{I}(s)\le & \exp\left(-\pi\lambda_2 \frac{{( sP)}^{2/\alpha}}{\sinc(2/\alpha)}\right)\times \\ \label{CovPrbPHP} &\exp\left(\!-2\pi \lambda_1\left(\int_{D}^{\infty}\left(1-\exp\left(f(v) \right)\right)v\mathrm{d}v\right)\right)
\end{align}
where $f(v)=\int_{{v-D}}^{{v+D}}{\arccos\left(\frac{r^2+v^2-D^2}{2vr}\right)}\frac{2\lambda_2}{1+\frac{r^\alpha}{Ps}} r\mathrm{d}r$.
\end{theorem}
\begin{proof}  See Appendix~\ref{app: D}.
\end{proof}
For the same reason as the tightness of lower bound of Theorem~\ref{thm :nearhole} discussed in Remark~\ref{remark3_r3}, the above upper bound is also remarkably tight for a wide variety of scenarios. More details on the tightness are provided in the next subsection and in the numerical results section. 
%\begin{remark}\label{rem:4}
%The tightness of this upper bound for the coverage probability of a PHP will be validated in results and discussion section.
%This approach does not capture the overlap effect among the holes.
%However, the upper bound is fairly accurate since it can capture the effect of  interference distribution in close proximity of the typical receiver.  
%\end{remark}
\subsection{Ratio of the Proposed Upper and Lower Bounds}
To study the tightness of the proposed upper and lower bounds, we derive a tight approximation on the ratio of upper and lower bounds, and show that it is close to one.
\begin{prop} \label{RatioBnd} The ratio of the upper and lower bounds on the Laplace transforms derived in Theorems~\ref{IccThm3} and \ref{thm :nearhole} is $\frac{\ncalL_\nrmu (s)}{\ncalL_\nrml (s)}\approx $
\begin{align}%\nonumber
%&\frac{\ncalL_\nrmu (s)}{\ncalL_\nrml (s)}\approx \\   
 \label{RatioBnd0} &\int_D^{\infty}\exp\left[-2\pi\lambda_1 \int_{v_1}^{\infty}\left(1-\exp\left(f(v) \right)\right)v\mathrm{d}v\right]f_{V_1}(v_1)\nrmd v_1,
\end{align}
where ${\ncalL_\nrmu (s)}$ and ${\ncalL_\nrml (s)}$ denote the proposed upper and lower bounds, given by Theorem~\ref{IccThm3} and Theorem~\ref{thm :nearhole}, respectively. 
Further, $f(v)=\int_{{v-D}}^{{v+D}}{\arccos\left(\frac{r^2+v^2-D^2}{2vr}\right)}\frac{2\lambda_2}{1+\frac{r^\alpha}{sP}} r\mathrm{d}r$.
\end{prop}
\begin{proof}  See Appendix~\ref{app: E}.
\end{proof}
This approximation can be interpreted as the Laplace transform of interference power {\em removed} by all the holes except the closest hole from the homogeneous PPP $\Phi_2$ after ignoring the effect of overlaps. This ratio will be shown to be tight and close to one across wide range of parameters in the numerical results section. Next, we explore a few ways to incorporate the effect of overlaps in the holes that was ignored in the results derived in this section. 
%the derivation of the upper bound in Theorem~\ref{IccThm3}. 

%In the next section, we investigate the challenges corresponding to the exact analysis of the overlap effect.  
%We further study the effect of  overlap on our proposed approach for the performance analysis of a PHP using three different techniques. 

%\begin{remark}\label{tbon_ratio}
%The tightness of this approximation for the ratio of the proposed upper and lower bounds will be demonstrated later.  
%Note that upper-lower bound ratio approaches to one when both bounds are close to the actual results. 
%We will show that this ratio is close to one across wide range of parameters.% .
%\end{remark}

\section{Incorporating the Impact of Overlaps in the Proposed Approaches}
                     \begin{figure}[t!]
  \centering{
              \includegraphics[width=.65\linewidth]{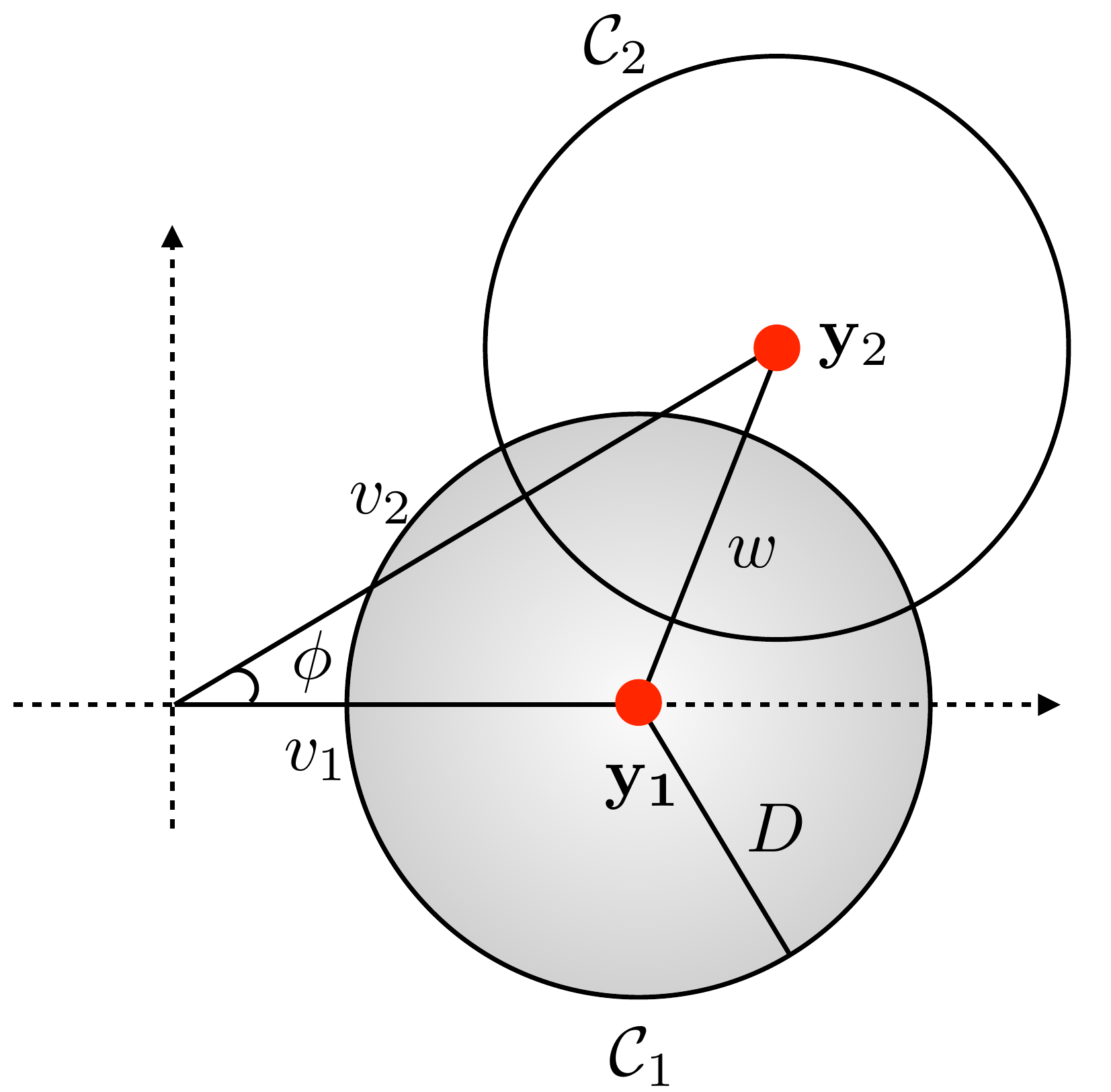}
              \caption{Illustration of the setup used in Theorem~\ref{Lemma2_nProof} where only two holes closest to the typical node are considered.}
\label{pdf_w}}
          \end{figure}

The key lower and upper bounds reported in Theorems~\ref{thm :nearhole} and \ref{IccThm3} in the previous section carefully circumvented the need for considering the overlaps in the holes explicitly. This was to maintain tractability. However, it is quite natural to wonder how much tractability is really lost if we try to incorporate the effect of the overlaps in the holes accurately. In this section, we address this question by deriving three results for the Laplace transform of interference that incorporate the effect of holes. In the first result, we generalize the lower bound derived in Theorem~\ref{thm :nearhole} by considering two nearest holes instead of a single hole. The overlap among the two holes is explicitly incorporated in the analysis. The setup is presented in Fig.~\ref{pdf_w}, where ${\ncalC_1}=\nbb(\nby_1,D)$ and  ${\ncalC_2}=\nbb(\nby_2,D)$ denote the first and the second closest holes to the typical receiver, respectively. The angle between $\nby_1$ and $\nby_2$ is denoted by $\phi$. The interference field in this case is modeled by $(\Phi_2 \cap \{\ncalC_1 \cup \ncalC_2\}^c) \supset \Psi$, which overestimates the interference power and hence leads to a lower bound on the Laplace transform of interference. Before going into the main result, we first need to evaluate the joint PDF of the distances between the first and second closest holes to the origin, which are denoted by random variables $V_1$ and $V_2$, respectively. Using the same arguments as in Lemma~\ref{Clos_Dislem}, the joint PDF can be derived as~\cite{haenggi2005distances}
\begin{align} \nonumber
f_{V_1V_2}({v_1,v_2})&=f_{V_2}({v_2|v_1})f_{V_1}({v_1})\\  \label{pdf_joint2}
&=(2\pi\lambda_1)^2 v_1v_2 \exp(-\pi\lambda_1 (v_2^2-D^2)).
\end{align}
Using this distribution, the Laplace transform of interference for this case is derived next.
%The Laplace transform of interference by considering two closest holes is presented in Lemma~\ref{Lemma2_nProof} as follows
\begin{theorem}[{New Lower Bound 2}]
\label{Lemma2_nProof}
The Laplace transform of interference experienced by a typical node of a PHP $\Psi$ is lower bounded by 
\begin{align}\nonumber
\ncalL_{{I}}(s)&\geq
\exp\left(-\pi\lambda_2 \frac{{(sP)}^{2/\alpha}}{\mathrm{sinc}(2/\alpha)}\right) \times \\\nonumber& \Bigg(\frac{1}{2\pi } \int_D^{\infty}\!\int_{v_1}^{\infty}\!\!\int_{-\pi}^{\pi}\exp\left(\int_{v_1-D}^{v_1+D}\!\!\frac{2\pi \lambda_{c1}(r)}{1+\frac{r^{\alpha}}{sP}} r \nrmd r\right)\\\nonumber &
\exp\left(\int_{v_2-D}^{v_2+D}\!\!\frac{2\pi \lambda_{c2}(r)}{1+\frac{r^{\alpha}}{sP}} r \nrmd r\right)
\exp\left(-\lambda_2\ncalB(v_1,v_2, \phi)\right)
\times \\   \label{LtI asym_07}  & f_{V_1V_2}({v_1,v_2})\nrmd \phi \nrmd v_2 \nrmd v_1\Bigg)
\end{align}
where $\lambda_{ci}(r)= \frac{\lambda_2}{\pi}{\arccos\left(\frac{r^2+{v_i}^2-D^2}{2{v_i}r}\right)}$, for $i=1,2$ and $\ncalB(v_1,v_2, \phi)$ is given by~(\ref{ovl_equav1v2ph}).
\end{theorem}
\begin{proof}  
The key idea behind this proof is to consider $(\Phi_2 \cap \{\nbb(\nby_1,D) \cup \nbb(\nby_2,D)\}^c) \supset \Psi$ as the interference field, which clearly overestimates the interference and hence leads to the lower bound. Complete proof is provided in Appendix~\ref{app: F}.
\end{proof}
%The Laplace transform of interference of the PHP in which two closest holes are considered leads to a challenging analytical problem.
%It is evident that increasing the number of holes along with considering the exact effect of overlap further increases the complexity of the problem.
                     \begin{figure}[t!]
  \centering{
              \includegraphics[width=.65\linewidth]{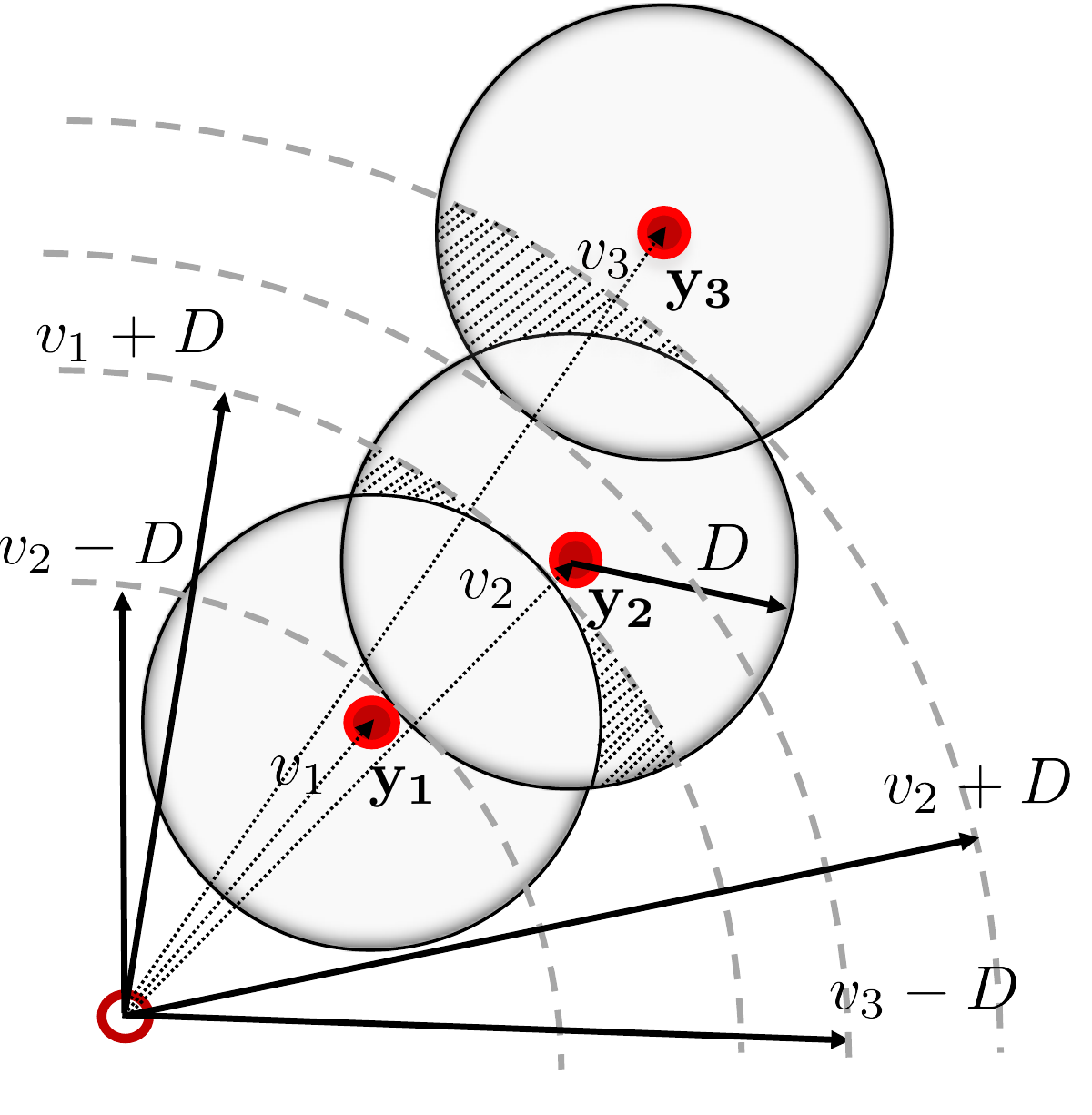}
              \caption{Illustration of the setup used in Theorem~\ref{Lemma2_kClosestBnd} where $k$ holes closest to the typical node are considered.}
\label{ovlp_3}}
          \end{figure}
As evident from the above result, incorporating the  effect of overlap, even among two holes, results in a significantly more complex expression compared to the bounds presented in Theorems~\ref{thm :nearhole} and \ref{IccThm3}. This shows that incorporating the {\em exact} effect of overlaps does indeed lead to a significant loss in tractability. Therefore, instead of trying to incorporate the exact effect of overlaps, we now propose a new procedure for bounding the overlap area, which allows us to derive a lower bound on the Laplace transform while considering multiple holes in the interference field. This tightens the result of Theorem~\ref{thm :nearhole}, where only one hole was considered. In particular, we consider $k$ closest holes from the typical receiver, as shown in Fig. \ref{ovlp_3}. In order to claim the result as a bound, we bound the union of $k$-closest holes, i.e., $\cup_{i=1}^{k}{\ncalC_i}=\cup_{i=1}^{k}  \nbb(\nby_i,D)$, with $\Omega_{\rm o}={{\ncalC_1}\cup _{i=2}^k  \left\{\nbb(\nby_i,D)\bigcap \nbb^c(0,\max{(\|\nby_{i-1}\|+D,\|\nby_i\|-D)}\right\}}$, where $\Omega_{\rm o} \subset \cup_{i=1}^k{\ncalC_i}$.
The set $\Omega_{\rm o}$ carefully avoids the overlapping part and hence does not result in over-removal of the points from the baseline point process $\Phi_2$. The contribution of $i^{th}$ hole in $\Omega_{\rm o}$ is $\nbd(\nby_i,D)=$
\begin{align}\label{eq:sect2:dy-D}
%\nbd(\nby_i,D)&=
\left\{\nbb(\nby_i,D)\bigcap \nbb^c(0,\max{(\|\nby_{i-1}\|+D,\|\nby_i\|-D)}\right\}.
\end{align}
In Fig. \ref{ovlp_3}, the set $\nbd(\nby_i,D)$ corresponds to the unshaded part of a hole that does not overlap with the 
other holes. Clearly, this approach results in the removal of less points from the baseline process $\Phi_2$ compared to a PHP, which leads to a lower bound on the Laplace transform of interference in a PHP. This lower bound is presented in the next theorem.
\begin{theorem}[{New Lower Bound 3}]
\label{Lemma2_kClosestBnd}
The Laplace transform of interference is bounded by
\begin{align}\nonumber 
&\ncalL_{{I}}(s)\geq
\exp\left(-\pi\lambda_2 \frac{{(sP)}^{2/\alpha}}{\mathrm{sinc}(2/\alpha)}\right) \times \\\nonumber& \Bigg( \int\!\int\!\!...\int_{D<v_1<v_2<...<v_k<\infty}\exp\left(\int_{v_1-D}^{v_1+D}\!\frac{2\pi \lambda_{c1}(r)}{1+\frac{r^{\alpha}}{sP}} r \nrmd r\right)\\ \nonumber&
\exp\left(\int_{\max({v_2-D,v_1+D})}^{v_2+D}\!\!\frac{2\pi \lambda_{c2}(r)}{1+\frac{r^{\alpha}}{sP}} r \nrmd r\right)...\\\nonumber&\exp\left(\int_{\max({v_k-D,v_{k-1}+D})}^{v_k+D}\!\!\frac{2\pi \lambda_{ck}(r)}{1+\frac{r^{\alpha}}{sP}} r \nrmd r\right)\times
\\ & f_{V_1V_2..V_k}({v_1,v_2,..,v_k})\nrmd v_1 \nrmd v_2...\nrmd v_k\Bigg)  \label{LtI KCL_01}
\end{align}
where $\lambda_{{\rm c}i}(r)= \frac{\lambda_2}{\pi}{\arccos\left(\frac{r^2+{v_i}^2-D^2}{2{v_i}r}\right)}$ for $i=1,2,...,k$,  and joint density function of distances of $V_1,V_2, ..., V_k$ is $f_{V_1V_2..V_k}({v_1,v_2,..,v_k})=(2\pi\lambda_1)^k v_1v_2...v_k \exp(-\pi\lambda_1 (v_k^2-D^2))$.
\end{theorem}
\begin{proof}  
See Appendix~\ref{app: H}.
\end{proof}
In the numerical results section, we will show that considering $k=2$, i.e., only the two closest holes from the typical point, results in a remarkably tight bound. Note that for $k=2$, the expression is also fairly tractable. For brevity, that expression is not stated separately.

Now, in Theorems~\ref{Lemma2_nProof},~\ref{Lemma2_kClosestBnd}, we have tried to handle the overlaps in such a way that the resulting expressions: (i) can be claimed as lower bounds to the  Laplace transform, and (ii) tighten the lower bound provided by Theorem~\ref{thm :nearhole}. One last question that we address before concluding this section is whether it is possible to handle the effect of overlaps in the {\em average sense}. For this, we revisit the upper bound of Theorem~\ref{IccThm3}, which was derived by carving out holes from the baseline process $\Phi_2$ individually without caring about the overlaps between them. This led to the possible removal of some points multiple times, thereby leading to an underestimation of interference. We compensate this {\em over-removal}, by rescaling the second term of Theorem~\ref{IccThm3}, which is the one that captures the effect of removing points from $\Phi_2$. The rescaling term is derived by estimating the average pairwise overlap area between circles. More details of the approach are provided in the proof in Appendix~\ref{app: G}. Note that the resulting expression in this case is the same as that of Theorem~\ref{IccThm3}, except a scaling factor of $1-{\min(\frac{{\lambda_1{\pi}D^2}}{2},\frac{1}{2})}$ that appears in the second term. The key downside to this approach compared to Theorem~\ref{IccThm3} is that it results in an {\em approximation} unlike Theorem~\ref{IccThm3}, where the result was shown to be a {\em bound}. 
\begin{prop} [{New Approximation}]\label{IccThm4} The Laplace transform of interference at a typical point in a PHP can be approximated as  $\ncalL_{I}(s)
\simeq$
\begin{align} \nonumber
%\ncalL_{I}(s)
%\simeq 
&\exp\left(-\pi\lambda_2 \frac{{(sP)}^{2/\alpha}}{\mathrm{sinc}(2/\alpha)}\right) 
\exp\bigg[-2\pi\lambda_1\int_{D}^{\infty} \bigg(1-\exp\bigg(
\\  \label{CovPrbPHP_2}  &{f(v)\times \bigg(1-  \min\bigg(\frac{{\lambda_1{\pi}D^2}}{2},\frac{1}{2}\bigg)\bigg)}\bigg)\bigg)v\mathrm{d}v\bigg]
\end{align}
where $f(v)=\int_{{v-D}}^{{v+D}}{\arccos\left(\frac{r^2+v^2-D^2}{2vr}\right)}\frac{2\lambda_2}{1+\frac{r^\alpha}{Ps}} r\mathrm{d}r$.
\end{prop}
\begin{proof}  See Appendix~\ref{app: G}.
\end{proof}
%\begin{remark}\label{remark6}
Comparison of this result with the bounds derived in this paper shows that handling overlaps in the average sense may not work particularly well, especially when the overlaps are significant. This is because such results do not capture the local neighborhood of the typical point as carefully as the bounds derived in this paper. We now move on the numerical results section, where more such insights about the relative accuracy of bounds and approximations are presented.
%\end{remark}

\begin{figure*}[t!]
\hfill
\subfigure[]{\includegraphics[width=.23\linewidth, height=0.23\linewidth]{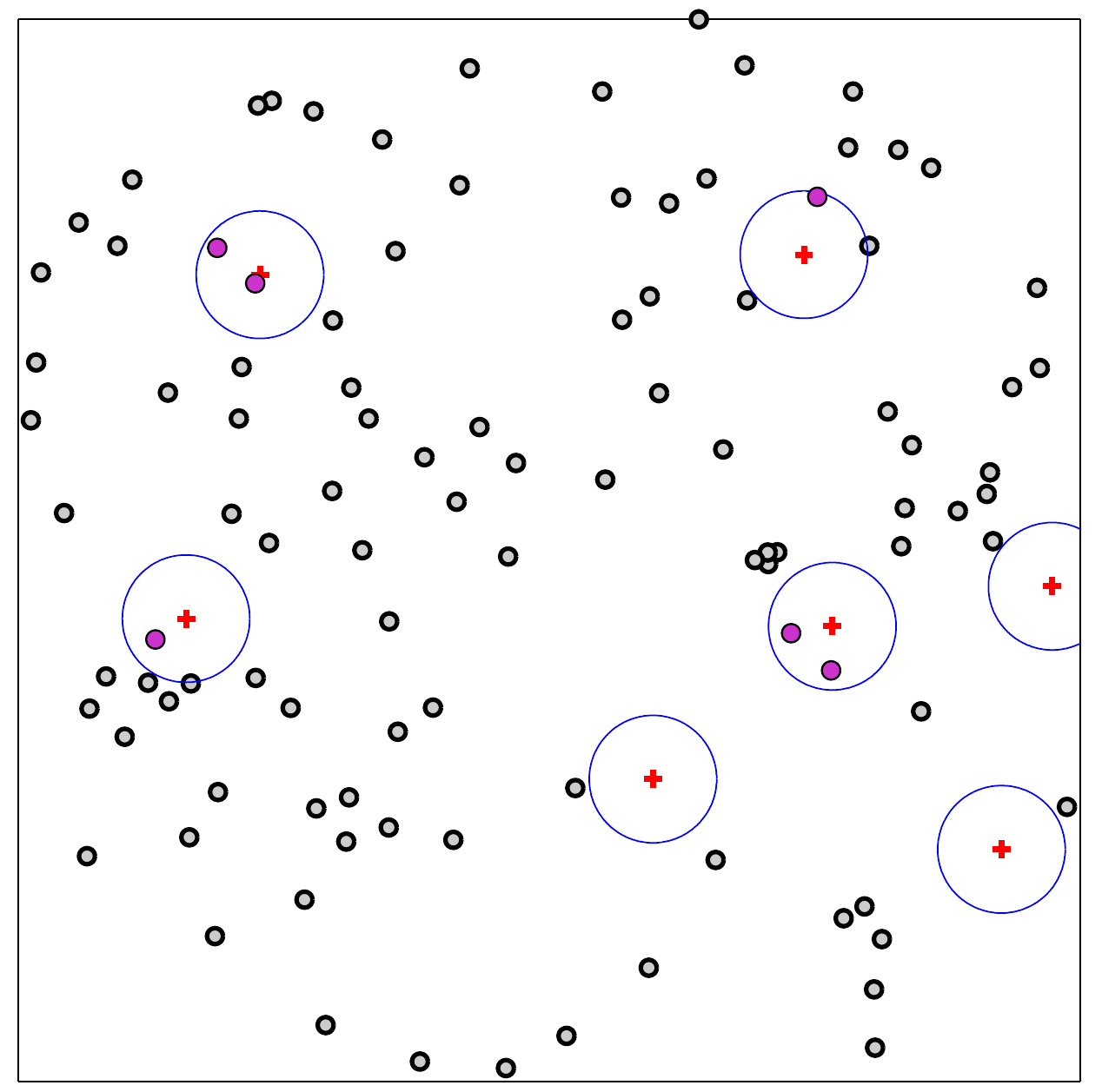}}
\label{netreal1}
\hfill
\subfigure[]{\includegraphics[width=.23\linewidth, height=0.23\linewidth]{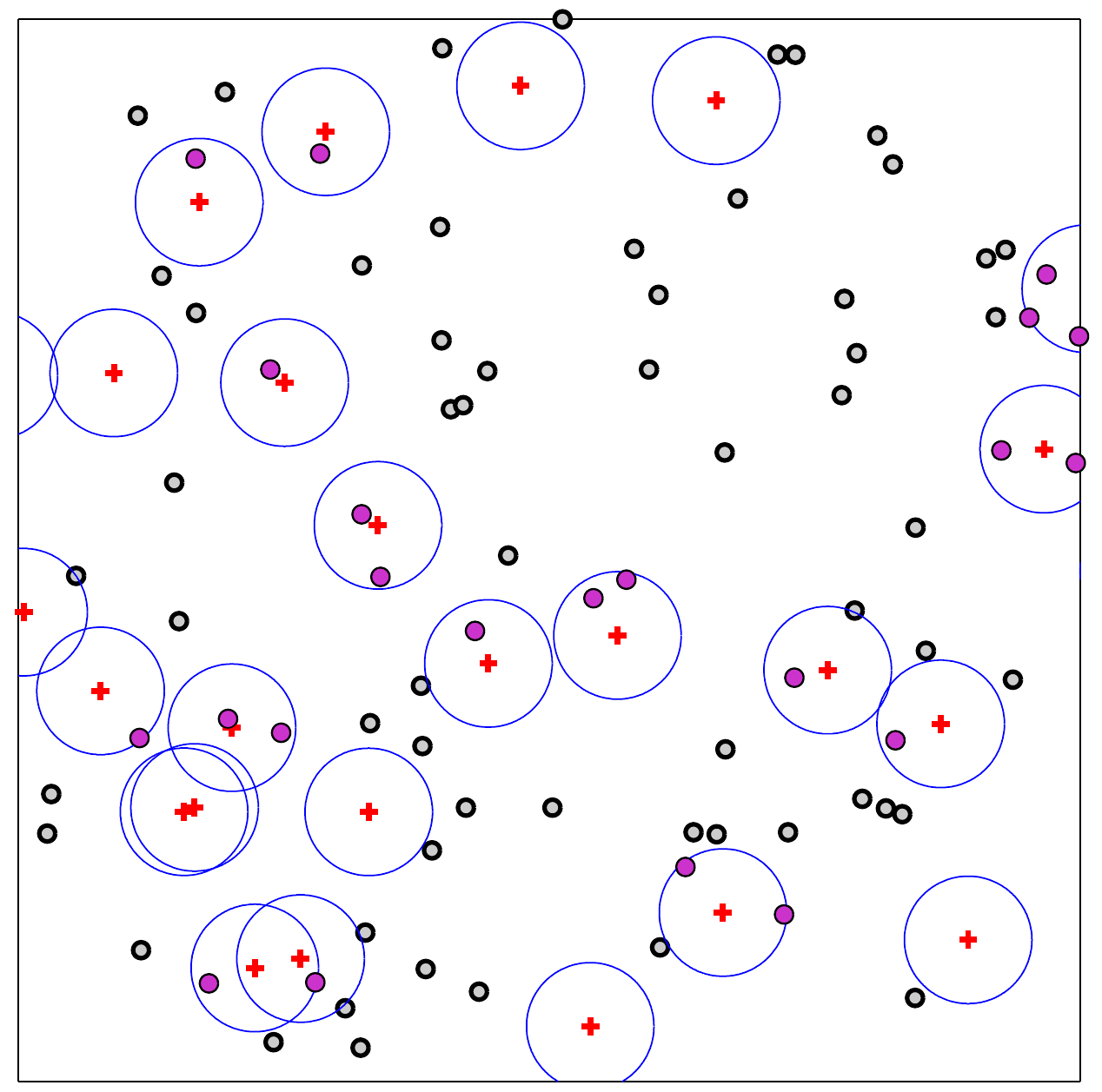}}
\label{netreal2}
\hfill
\hfill
\subfigure[]{\includegraphics[width=.23\linewidth, height=0.23\linewidth]{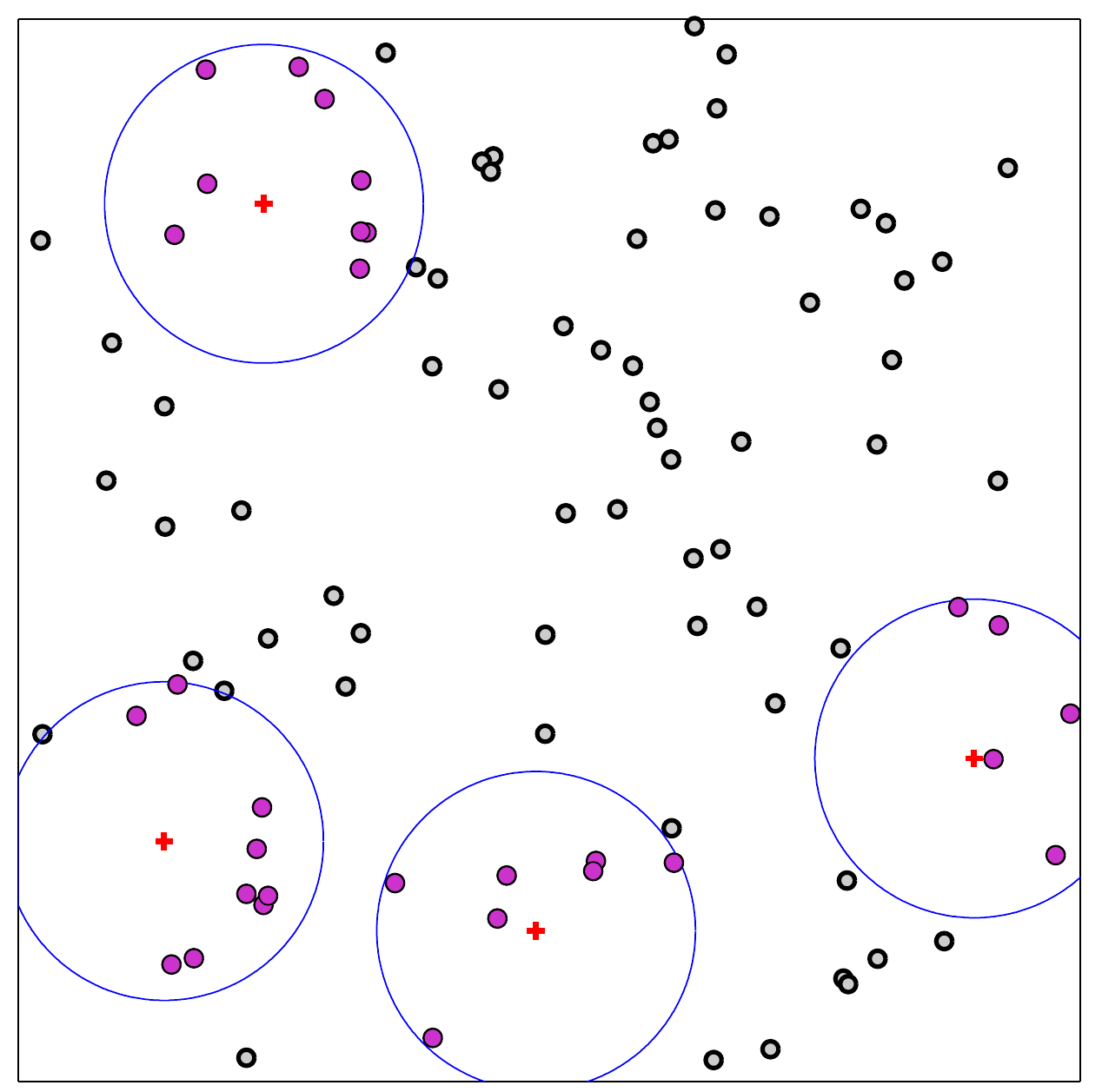}}
\label{netreal3}
\hfill
\subfigure[]{\includegraphics[width=.23\linewidth, height=0.23\linewidth]{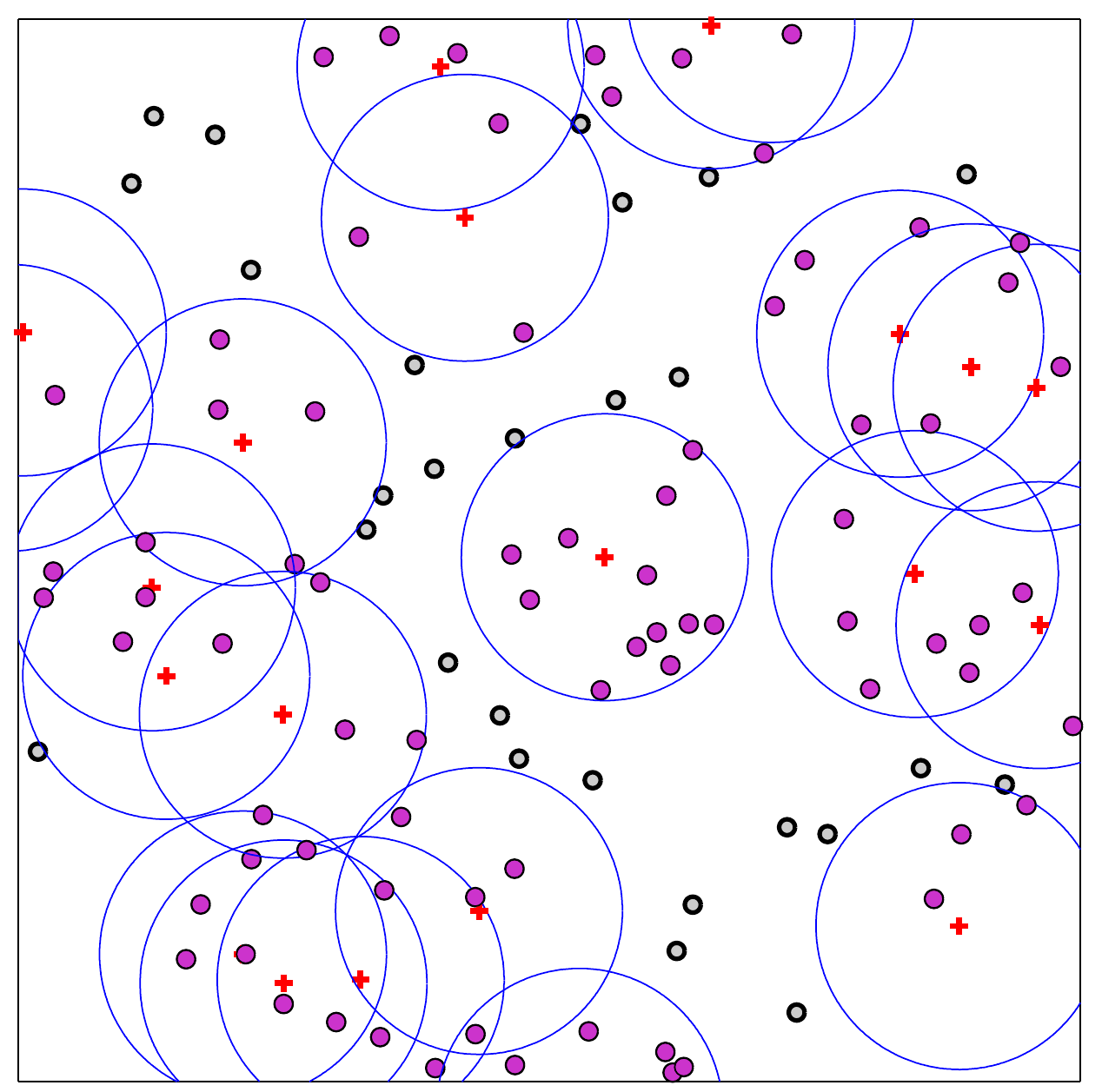}}
\hfill
 \caption{The PHP network model (a) First configuration: \LDSH, (b) Second configuration: \HDSH, (c) Third configuration: \LDLH,   (d)  Fourth configuration: \HDLH.}
 \label{netreal4}
\end{figure*} 

          \begin{figure}[t!]
  \centering{
              \includegraphics[width=.9\linewidth]{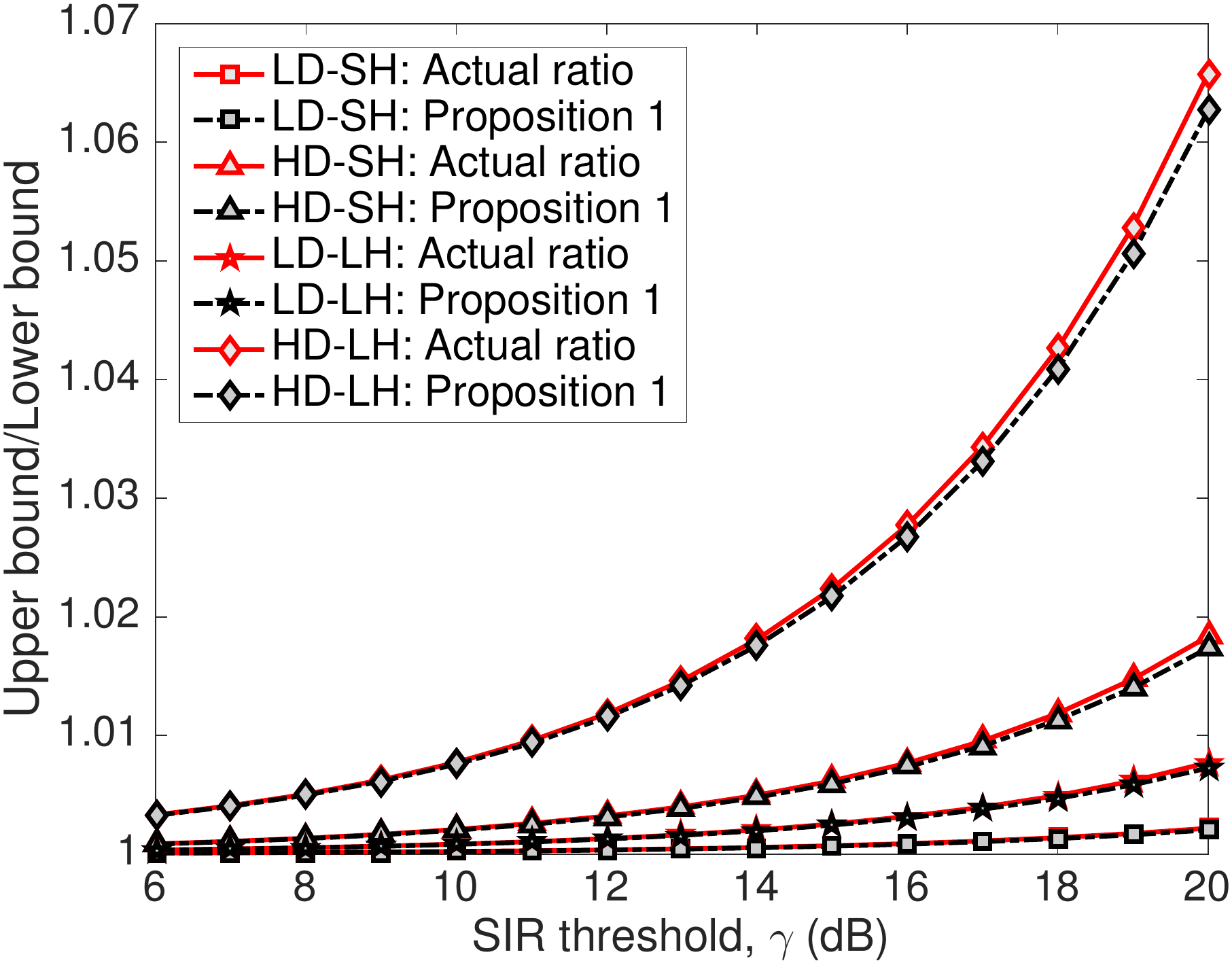}
              \caption{Ratio of the proposed upper and lower bounds derived in Theorems~\ref{IccThm3} and  \ref{thm :nearhole}, respectively.}
\label{UpLwBnd1}}
          \end{figure}   

\begin{figure}[t!]
  \centering{
              \includegraphics[width=.9\linewidth]{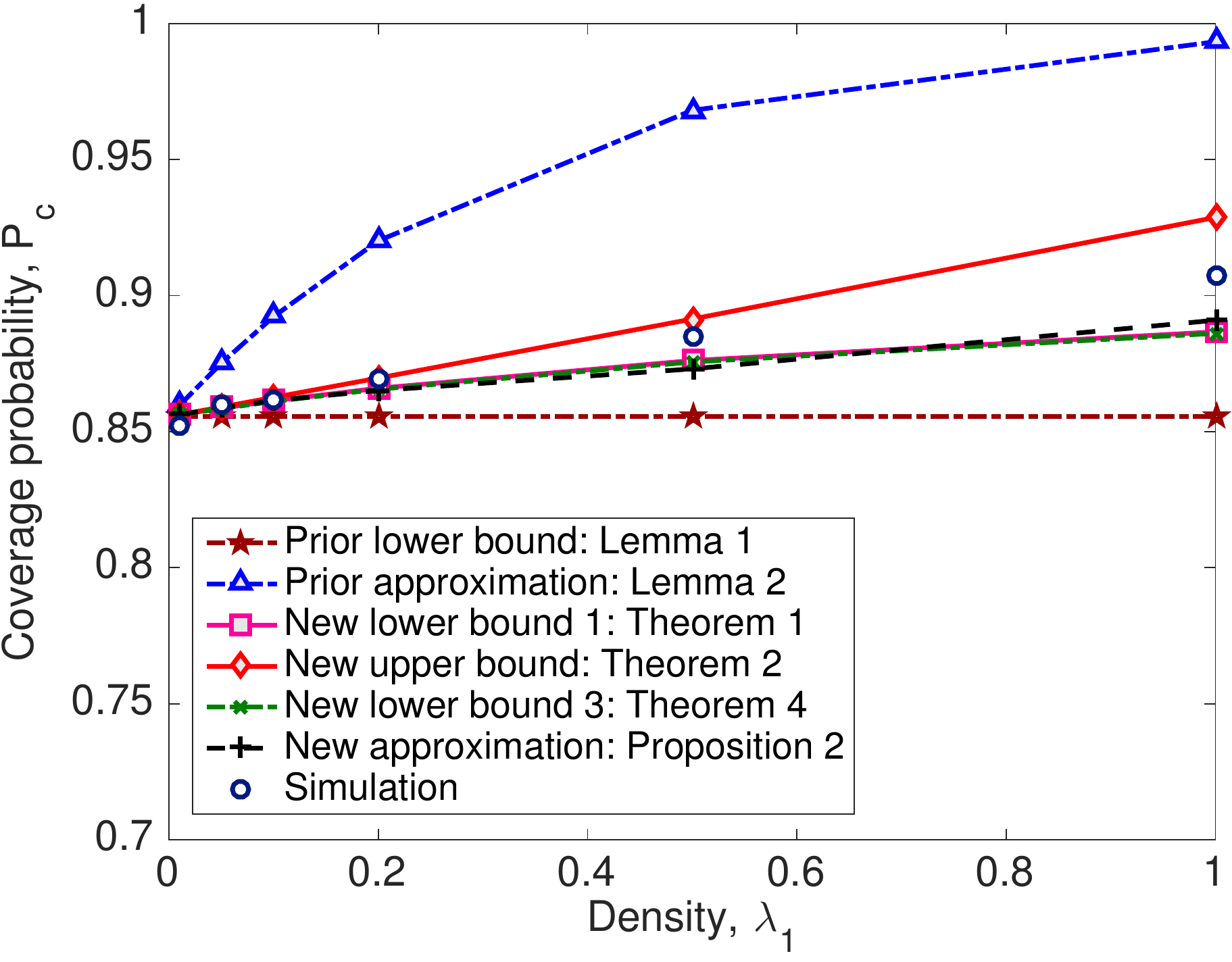}
\caption{Analytical and simulation results for the the coverage probability as a function of the density $\lambda_1$ ($D=1$).}
\label{LISCDfixLchg10}}
          \end{figure}
\begin{figure}[t!]
  \centering{
              \includegraphics[width=.9\linewidth]{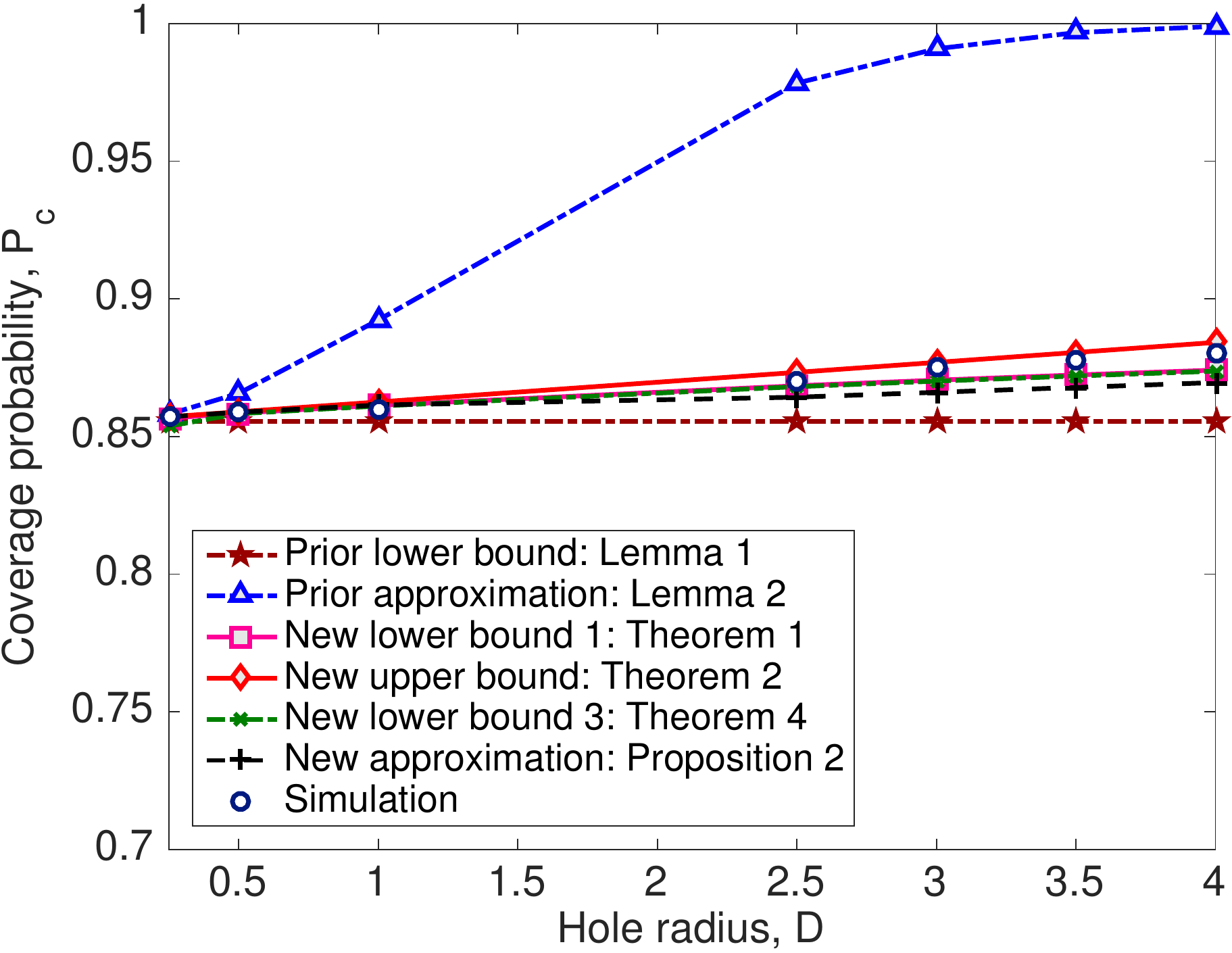}
              \caption{Analytical and simulation results for the coverage probability as a function of hole radius $D$ ($\lambda_1=0.1$).}
\label{LISClfixDchg05}}
          \end{figure}

                     \begin{figure}[t!]
  \centering{
              \includegraphics[width=.9\linewidth]{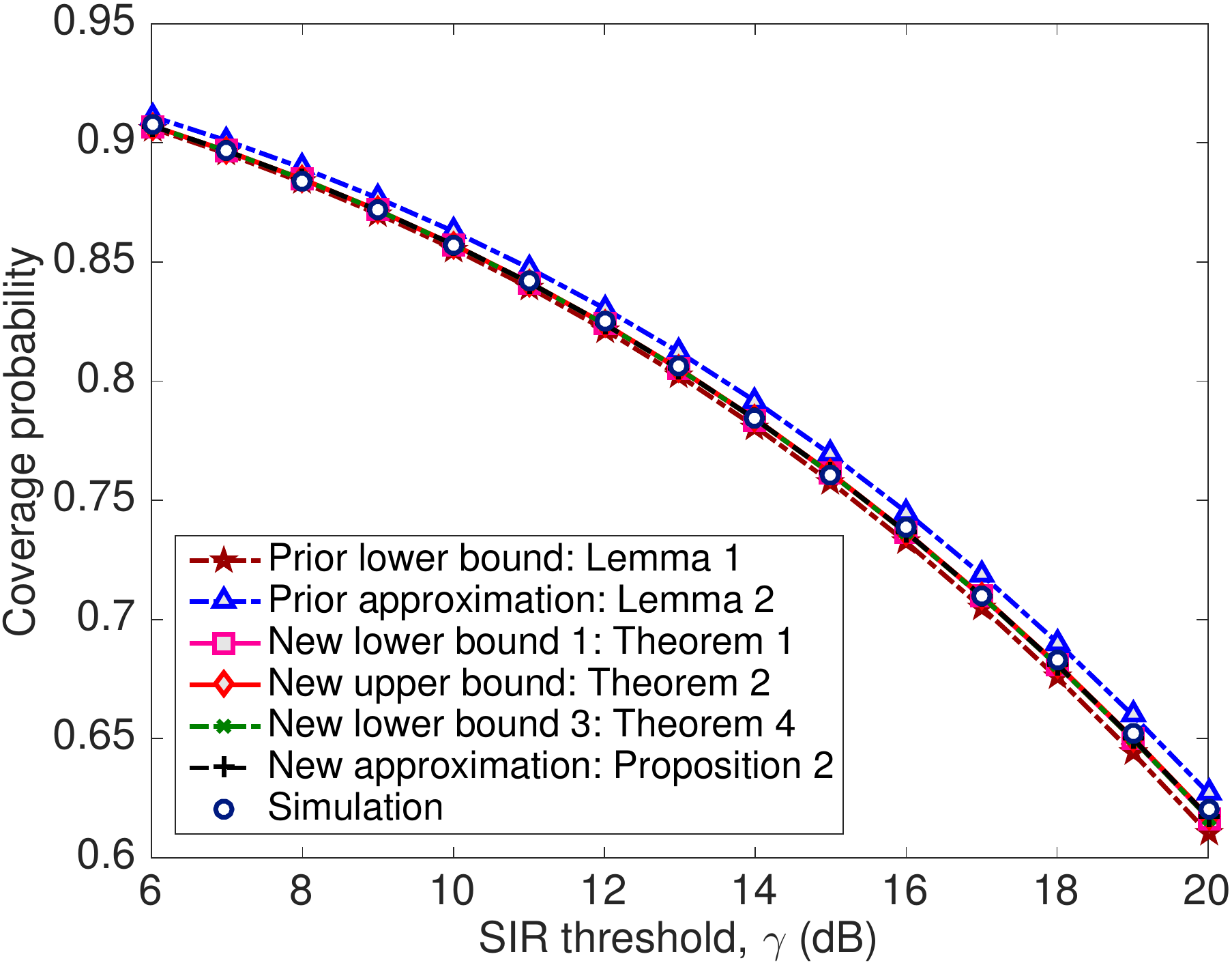}
              \caption{Analytical and simulation results for the coverage probability in \LDSH~case ($\lambda_1=0.05$ and $D=0.6$).}
\label{LISC_C1_001}}
          \end{figure}
                     \begin{figure}[t!]
  \centering{
              \includegraphics[width=.9\linewidth]{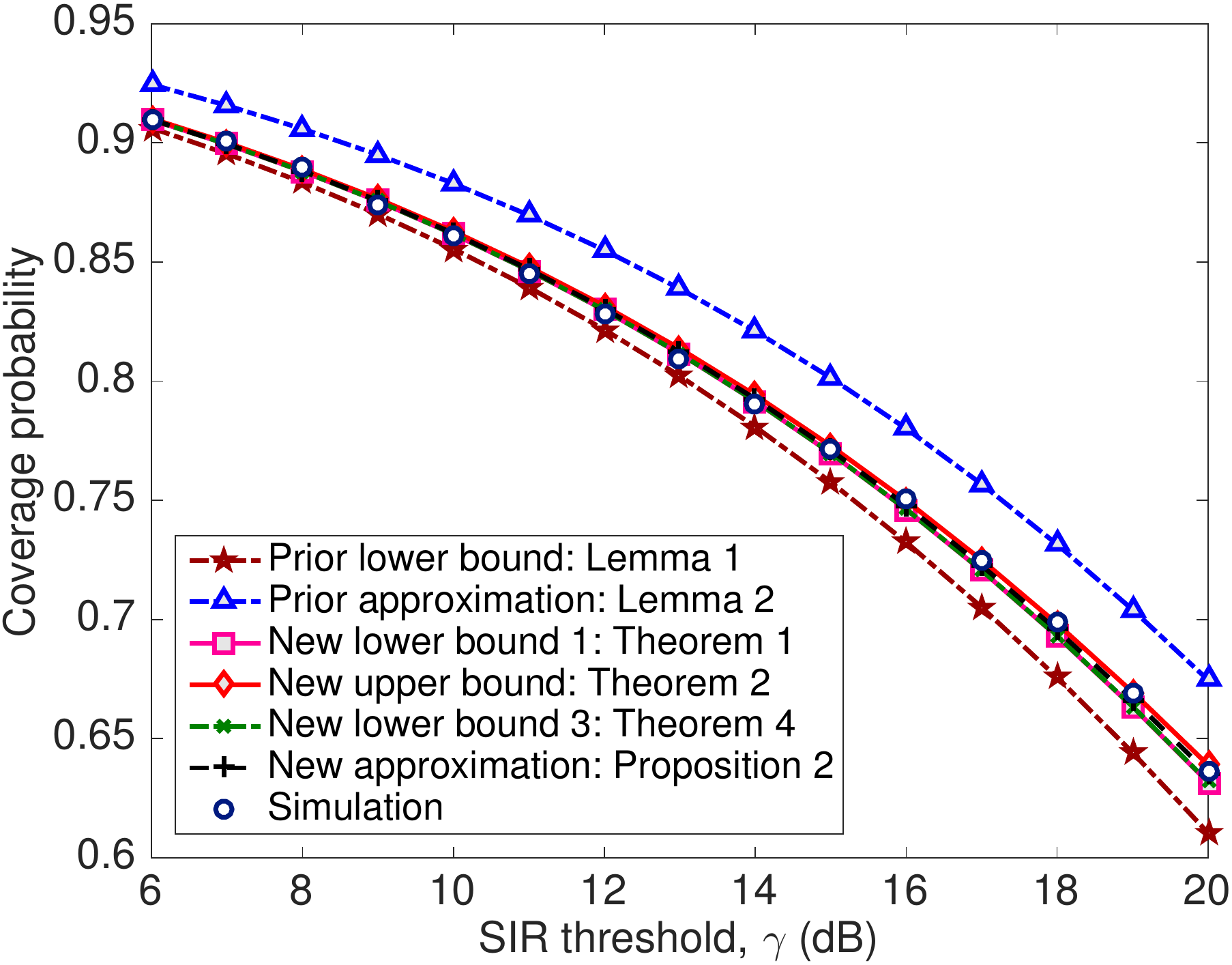}
              \caption{Analytical and simulation results for the coverage probability in \HDSH~case ($\lambda_1=0.2$ and $D=0.6$).}
\label{LISC_C2_002}}
          \end{figure}

                     \begin{figure}[t!]
  \centering{
              \includegraphics[width=.9\linewidth]{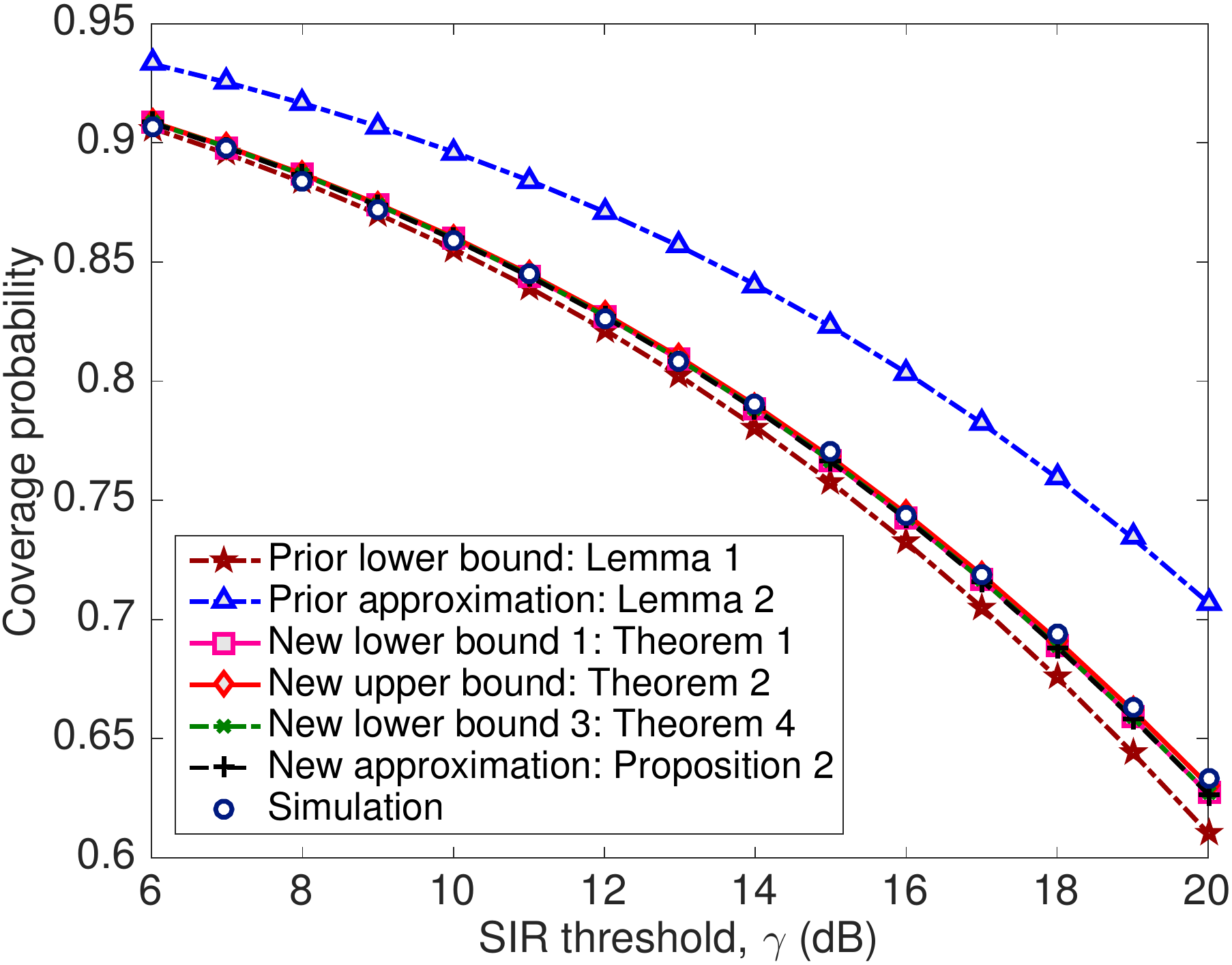}
              \caption{Analytical and simulation results for the coverage probability in \LDLH~case ($\lambda_1=0.05$ and $D=1.5$).}
\label{LISC_C3_003}}
          \end{figure}
                     \begin{figure}[t!]
  \centering{
              \includegraphics[width=.9\linewidth]{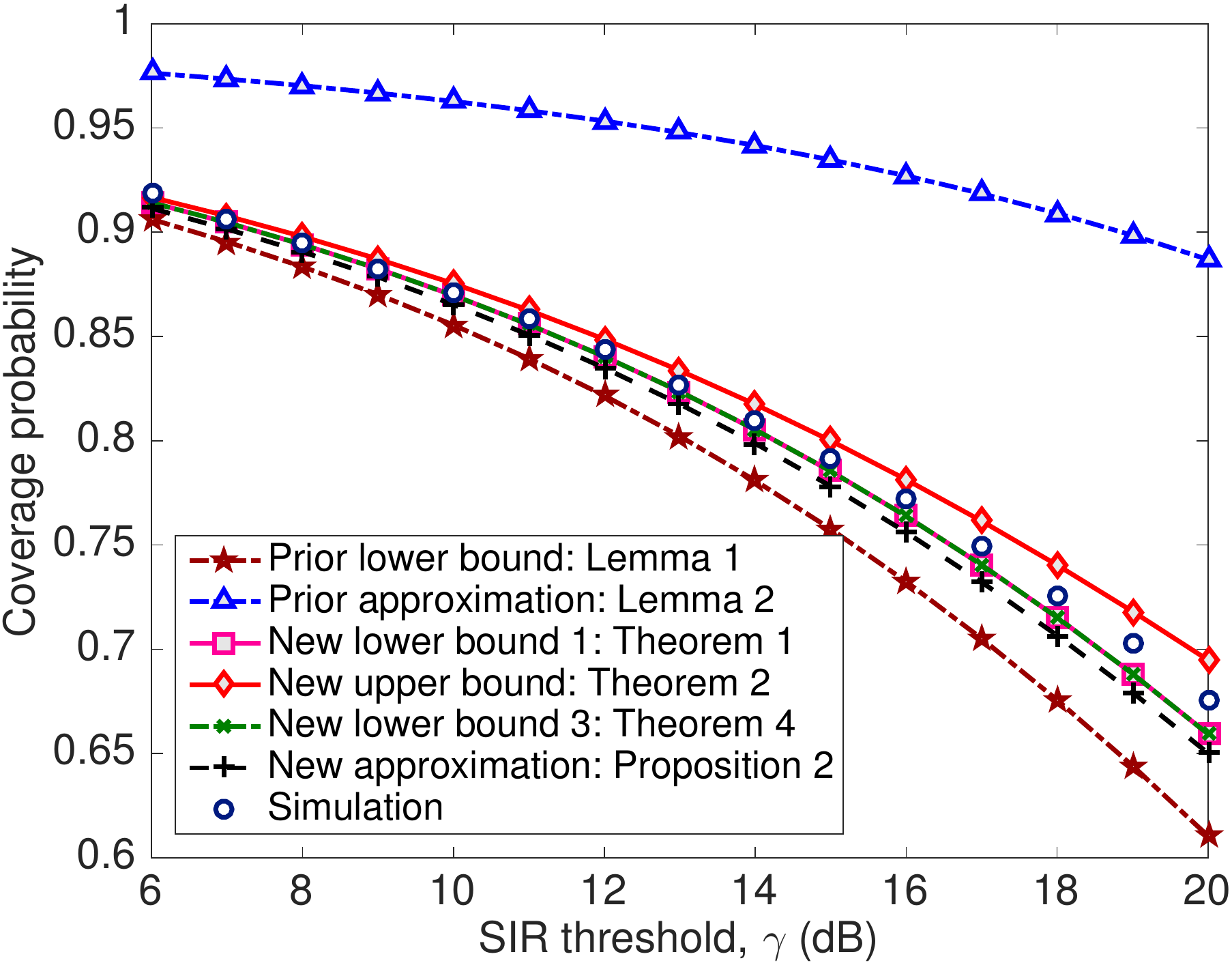}
              \caption{Analytical and simulation results for the coverage probability in \HDLH~case ($\lambda_1=0.2$ and $D=1.5$).}
\label{LISC_C4_004}}
          \end{figure}      
  
\section{Numerical Results and Discussion}
\label{Discuss_PHP}
There are three main system parameters that determine the interference experienced by the typical point in a PHP: density $\lambda_1$ of the holes, density $\lambda_2$ of the baseline PPP, and the radii $D$ of the holes. Based on the relative values of these parameters, we identify four main network {\em configurations}, which are illustrated in Fig. \ref{netreal4}. 
We define the possible configurations as
%\begin{enumerate}
\LDSH: configuration with low density of holes and small holes;  
%(small $\lambda_1$ and small $D$)
\HDSH: configuration with high density of holes and small holes; 
% (large $\lambda_1$ and small $D$).
\LDLH: configuration with low density of holes and large holes;  
%(small $\lambda_1$ and large $D$).
\HDLH: configuration with high density of holes and large holes. 
%(large $\lambda_1$ and large $D$).
%\end{enumerate}
Clearly, the configuration where the holes are small and sparse (LD-SH case) is more benign than the configuration in which the holes are small and dense (HD-SH). Similarly, the  configuration where holes are large and sparse (LD-LH case) is more benign than the configuration where the holes are both large and dense (HD-LH case). Therefore, the result that works well in the HD-SH configuration is expected to work in the LD-SH configuration as well. The same is true for LD-LH and HD-LH cases.

Simulations are performed over circular region with radius $40 m$ and results are averaged over $5\times10^4$ iterations. 
Unless otherwise specified, we set the network parameters as follows: $\lambda_2=1$, $\alpha=4$, $P=1$, $r_0=0.1$, $\gamma=10$ dB. 
We compare our proposed bounds and the new approximation with the first-order statistic approximation given by Lemma~\ref{thm :appr}, 
and the PPP-based bound given by Lemma~\ref{thm :lowbound} where $\Psi$ is approximated by $\Phi_2$.
Before going into more details comparisons, we  demonstrate the tightness of the lower and upper bounds derived in Theorems \ref{thm :nearhole} and~\ref{IccThm3} by plotting their ratio and its approximation (derived in Proposition~\ref{RatioBnd}) in Fig. \ref{UpLwBnd1}. In addition to validating the tightness of the approximation given by Proposition~\ref{RatioBnd}, this result shows that the ratio in all cases of interest is close to one, which demonstrates the tightness of both the bounds. Note that, as expected, the ratio is comparatively higher when the holes are large and dense (HD-LH case). 

We now compare the proposed bounds and approximations with the numerical results and the known approaches in terms of coverage probability. 
As demonstrated in \eqref{eq:Pcdef}, the coverage probability for our setup is simply the Laplace transform of interference evaluated at $s=\frac{\gamma r_0^\alpha}{P}$. Note that when we substitute $s=\frac{\gamma r_0^\alpha}{P}$ in any of the Laplace transform expressions derived in this paper, we notice that the resulting expression is independent of $P$. This is expected, because the $\mathtt{SIR}(r_0)$ expression given by \eqref{SIR} is indeed independent of $P$.

In Fig. \ref{LISCDfixLchg10}, we plot the coverage probability of a typical receiver as a function of hole density $\lambda_1$ assuming all other parameters are fixed (we assume $D=1$). Small values of $\lambda_1$ result in LD-SH configuration, whereas high values result in the HD-SH configuration. In Fig. \ref{LISClfixDchg05}, we conduct the same study but instead of varying $\lambda_1$, we now vary the hole radius $D$ while fixing $\lambda_1 = 0.1$. The low values of $D$ result in  HD-SH configuration, whereas the high values result in HD-LH configuration. Comparison of the proposed results with the simulations reveal that all the bounds and approximations proposed in this paper are surprisingly tight, even for the extreme configuration like HD-LH, where the overlaps are significant. The lower bounds work even better. We also notice that the prior results (given by Lemmas~\ref{thm :lowbound} and~\ref{thm :appr}) deviate significantly when the overlaps in the holes are significant. In particular, the approximation given by Lemma~\ref{thm :appr} becomes very loose. As discussed in Remark~\ref{rem:1}, this is because its derivation involved independent thinning of the interference field, which distorts the local neighborhood of the typical receiver. 

\begin{figure}[t!]
  \centering{
              \includegraphics[width=.9\linewidth]{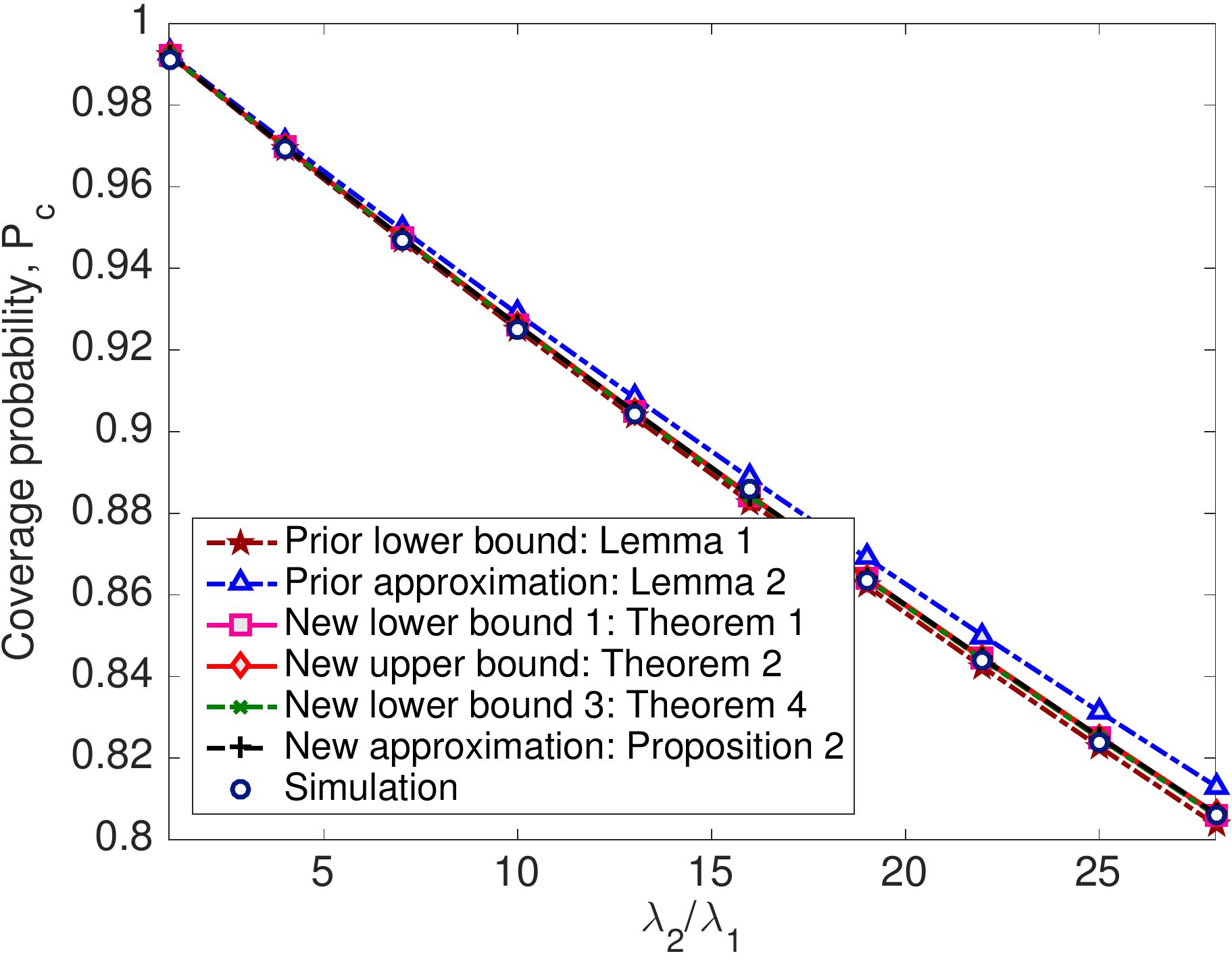}
              \caption{Analytical and simulation results for the coverage probabilities of the
PHP users as a function of $\lambda_2/\lambda_1$ under configuration \LDSH~($\lambda_1=0.05$ and $D=0.6$).}
\label{PoutSC_cfg1_1}}
          \end{figure}
\begin{figure}[t!]
  \centering{
              \includegraphics[width=.9\linewidth]{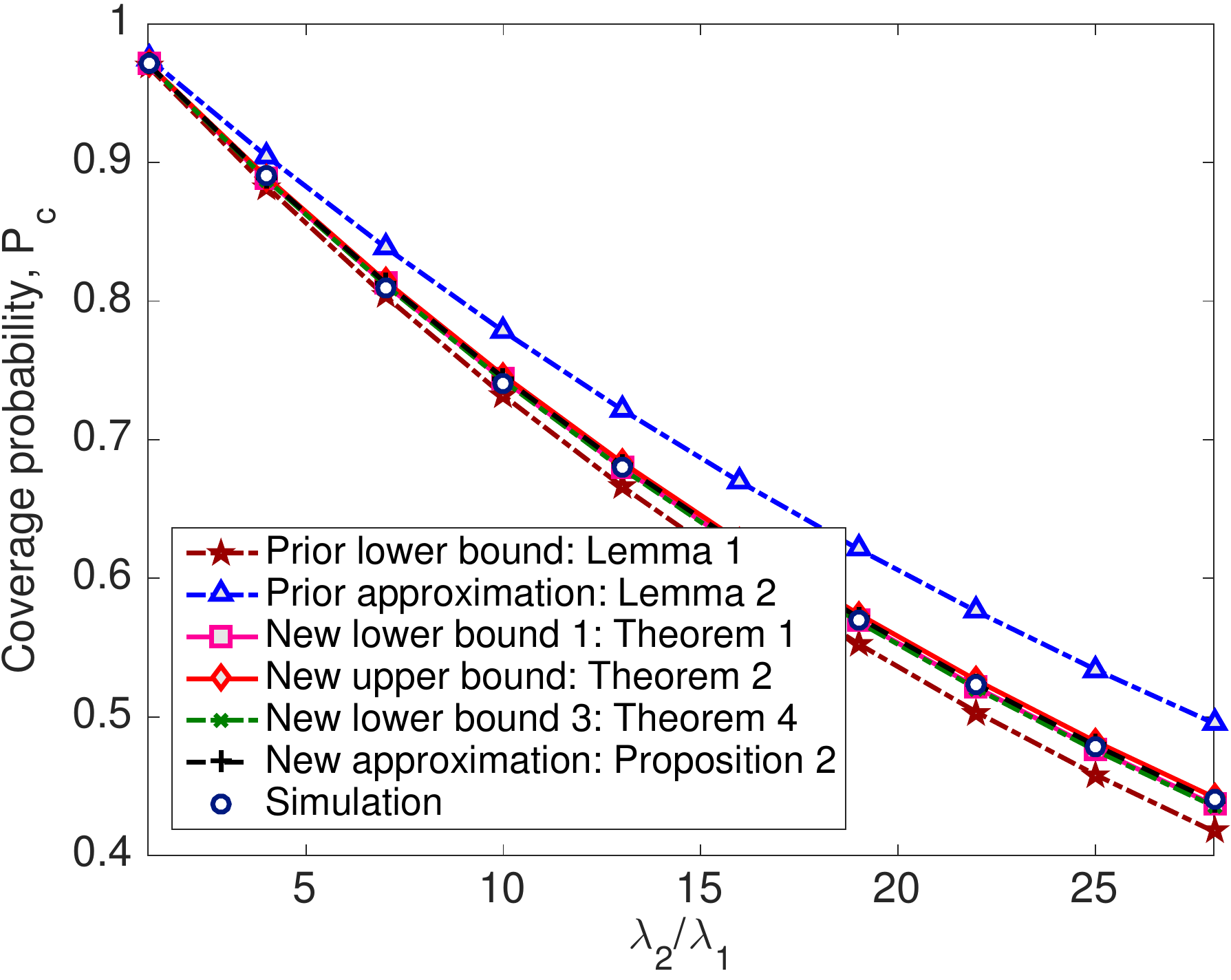}
              \caption{Analytical and simulation results for the coverage probabilities of the
PHP users as a function of $\lambda_2/\lambda_1$ under configuration \HDSH~($\lambda_1=0.2$ and $D=0.6$).}
\label{PoutSC_cfg2}}
          \end{figure}

\begin{figure}[t!]
  \centering{
              \includegraphics[width=.9\linewidth]{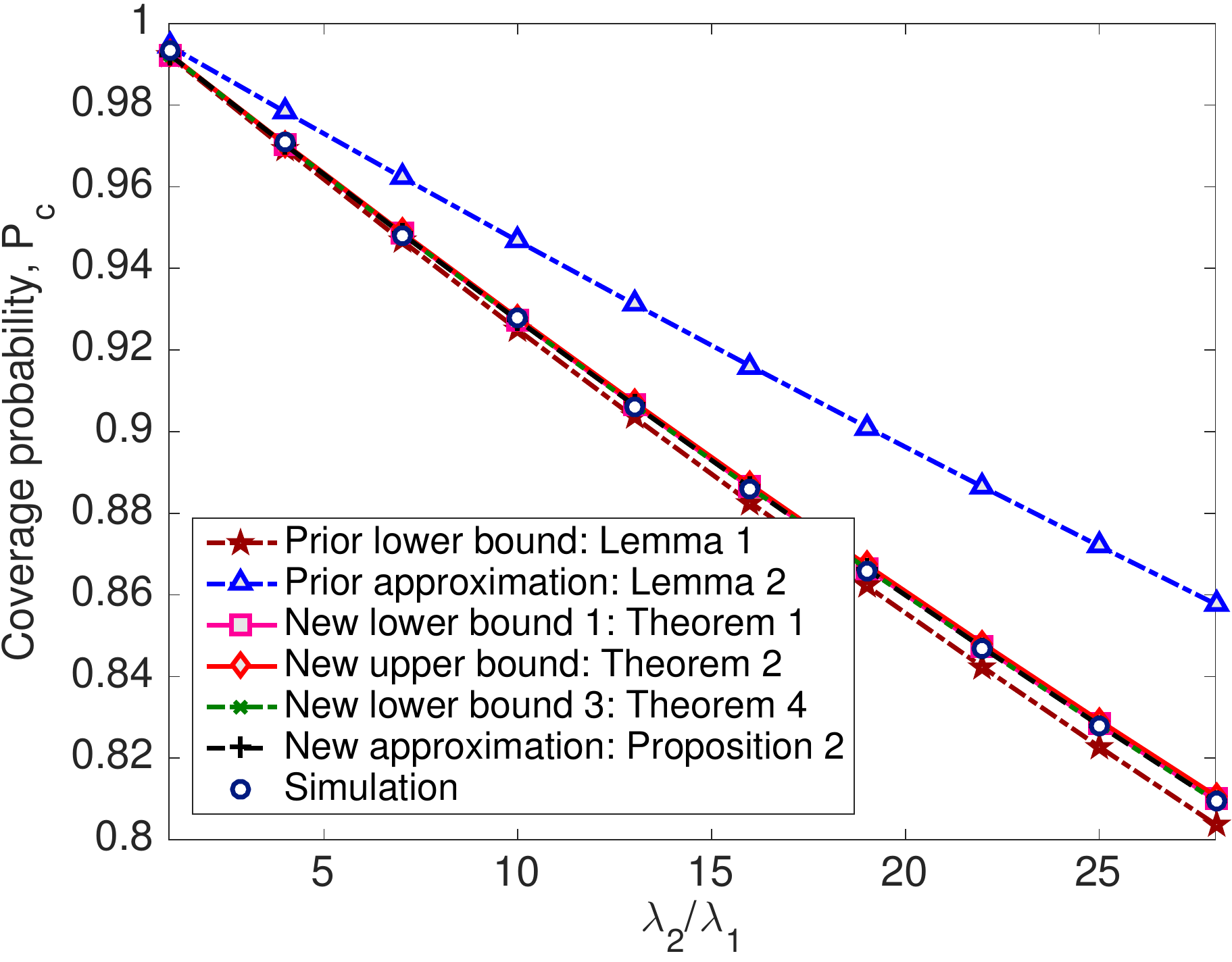}
              \caption{Analytical and simulation results for the coverage probabilities of the
PHP users as a function of $\lambda_2/\lambda_1$ under configuration \LDLH~($\lambda_1=0.05$ and $D=1.5$).}
\label{PoutSC_cfg3_3}}
          \end{figure} 
\begin{figure}[t!]
  \centering{
              \includegraphics[width=.9\linewidth]{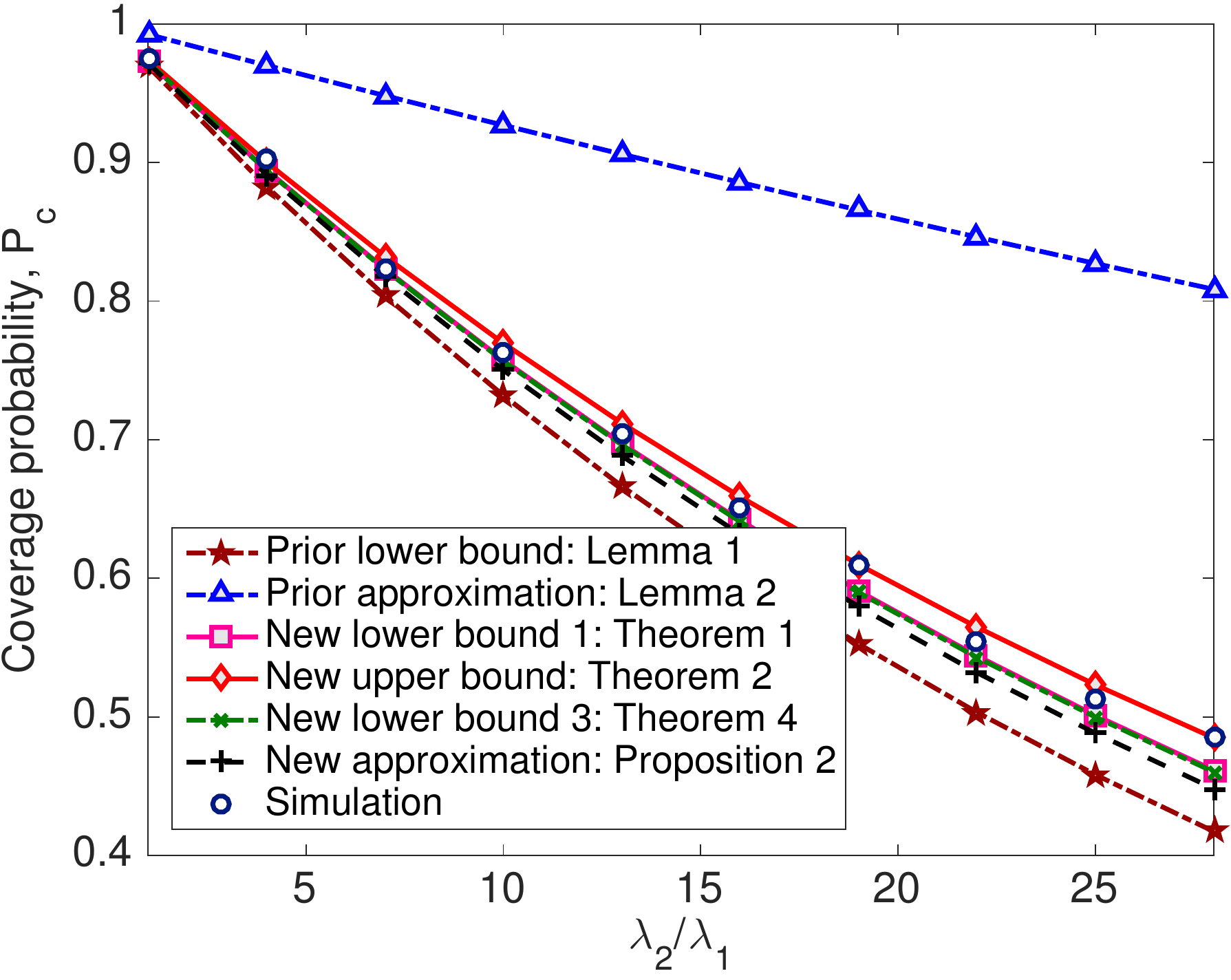}
              \caption{Analytical and simulation results for the coverage probabilities of the
PHP users as a function of $\lambda_2/\lambda_1$ under configuration \HDLH~($\lambda_1=0.2$ and $D=1.5$).}
\label{PoutSC_cfg4_4}}
          \end{figure}

We now plot the analytical and simulation results for the coverage probability of the typical receiver as a function of the $\sir$ threshold in the four possible configurations in Figs.~\ref{LISC_C1_001}--\ref{LISC_C4_004}. Fig. \ref{LISC_C1_001} shows the results for \LDSH~case with small holes of radius $D=0.6$ and low density of $\lambda_1=0.05$, while results for the \HDSH~case are shown in Fig. \ref{LISC_C2_002} with parameters $D=0.6$ and $\lambda_1=0.2$.
Further, Fig. \ref{LISC_C3_003} shows results for the \LDLH~case with large holes of radius $D=1.5$ and low density of $\lambda_1=0.05$, while the results for the \HDLH~case with parameters $D=1.5$ and $\lambda_1=0.2$ are shown in Fig. \ref{LISC_C4_004}. 
The plots again confirm the accuracy of our results and show that the first-order statistic approximation given by Lemma~\ref{thm :appr} is rather loose while our proposed bounds lead to tight upper and lower bounds in all cases. Interestingly, all the proposed results work so well that they are nearly indistinguishable in all configurations except the most extreme one of HD-LH (Fig. \ref{LISC_C4_004}). In this configuration, we first notice that both the lower bounds given by Theorems~\ref{thm :nearhole} and \ref{Lemma2_kClosestBnd} work equally well, which means considering even a single hole in the interference field accurately is good enough for the accurate characterization of interference. As expected, we also notice that the approximation derived by handling the overlaps in the average sense in Proposition~\ref{IccThm4} does not work better than the proposed bounds in the extreme configuration of HD-LH (Fig. \ref{LISC_C4_004}). This is because any {\em average-based} arguments do not necessarily capture the local neighborhood of the typical point as well as the bounds do.

For completeness, we also plot the analytical and simulation results for the coverage probability of the typical receiver in a PHP as a function of $\lambda_2/\lambda_1$ in the four configurations in Figs. \ref{PoutSC_cfg1_1}--Fig. \ref{PoutSC_cfg4_4}. In Figs. \ref{PoutSC_cfg1_1} and  \ref{PoutSC_cfg2}, design parameters are set in order to simulate \LDSH~and \HDSH~cases, respectively. In particular, we assume $D=0.6$ and $\lambda_1=\{0.05, 0.2\}$. Figs. \ref{PoutSC_cfg3_3} and  \ref{PoutSC_cfg4_4}  depict results for the \LDLH~ and \HDLH~cases. Here we consider $D=1.5$ and $\lambda_1=\{0.05, 0.2\}$. As was the case in the above results, our proposed lower and upper bounds provide a remarkably accurate characterization of coverage probability. This is because the local neighborhood of the typical node is carefully preserved while deriving these bounds. On the other hand, the prior results, in particular the first-order statistic approximation, leads to a fairly loose result. 

\section{Conclusions}
\label{Concl}
In this paper, we have focused on the accurate performance characterization of a typical user in a wireless network that is modeled as a PHP. This model is of particular interest in scenarios where interference management techniques introduce {\em spatial separation} among active transmitters in the form of holes or exclusion zones. In terms of the technical results, we have provided new easy-to-use provable lower and upper bounds on the Laplace transform of interference experienced by a typical user in a PHP. 
In addition to accurately characterizing
the interference power, these bounds immediately characterize the coverage probability of a typical user in the case where all the wireless links experience independent Rayleigh fading. Since the prior work has mostly focused on
reducing the PHP to a PPP either by ignoring the holes or by matching the PPP density to that of a PHP, to the best of our knowledge, the proposed bounds are the tightest known
bounds for the Laplace transform of interference in a PHP.
For the analysis, we proposed a new approach in which the holes are dissolved in such a way that a PHP is reduced to an equivalent (and more tractable) non-homogeneous PPP.
The key in deriving tight bounds was to preserve the local neighborhood around the typical point while simplifying the
far field to attain tractability. 
The tightness of the bounds is
demonstrated analytically as well as numerically by comparing with simulations and known approaches. These results have numerous applications in a variety of wireless networks where
interference management is performed by spatially separating the active links, such as in cognitive radio and D2D networks.

Since our main emphasis was on characterization of the interference power, we assumed that the serving transmitter for the typical receiver is located at a fixed distance. This corresponds to an {\em ad hoc} network scenario. Relaxation and generalization of this assumption is a fruitful area of future investigation. For instance, if the receiver of interest is a randomly chosen point in $\nbbR^2$ and its serving transmitter is its closest point in a PHP, this setup can be used to study the performance of a cellular network modeled as a PHP. Another direction of future work is the extension of the current model to study exclusion zones with different shapes and sized, e.g., circles with different radii. Finally, the holes in a PHP are driven by a PPP. Extending this to other point processes, such as a Mat\'{e}rn process, is another promising direction of future work.% 

\appendix
\subsection{Proof of Lemma~\ref{thm :lowbound}}
\label{app: A}
In this case, the PHP $\Psi$ is approximated by the baseline PPP $\Phi_2$, which reduces the problem to the well-studied problem of deriving Laplace transform of interference originating from the PPP field of interferers~\cite{haenggi2012stochastic}. For completeness, the sketch of the derivation is provided next.  
%We perform analysis on a typical receiver located at the origin which is made possible by Slivnyak's Theorem~\cite{haenggi2012stochastic}. 
%When the impact of holes is ignored, PHP reduces to a PPP of density $\lambda_2$. 
%Hence, the Laplace transform of interference is lower bounded by 
\begin{align}
\label{LtI asym_1}
&\ncalL_{{I}}(s)\geq\nbbE \left[\exp\left(-s{\!}\sum_{\nbx\in{\Phi_\mathrm{2}}}{\!}Ph_\nbx\|\nbx\|^{-\alpha} \right) \right]\\\nonumber &\stackrel{\text{(a)}}{=}
\exp\left({\!\!}-\lambda_2 \int_{\nbbR^2}\!1-\mathbb{E}_{h_\nbx}\bigg[\exp(-s P h_\nbx \|\nbx\|^{-\alpha})\bigg]\mathrm{d}\nbx \right)
\\\nonumber &\stackrel{\text{(b)}}{=}
\exp\left({\!\!}-\lambda_2 \int_{\nbbR^2}\frac{1}{1+\frac{\|\nbx\|^{\alpha}}{sP}}\mathrm{d}\nbx \right)
\stackrel{\text{(c)}}{=}\exp\left[-\pi\lambda_2 \frac{(sP)^{2/\alpha}}{\mathrm{sinc}(2/\alpha)}\right]
\end{align}
where (a) follows from the PGFL of a PPP~\cite{chiu2013stochastic}, (b) from  $h_\nbx \sim \exp(1)$, and (c) using standard machinery, where the integral is first converted form Cartesian to polar coordinates and the closed form expression follows from the properties of the Gamma function~\cite[Appendix B]{dhillon2012modeling}.

\subsection{Proof of Lemma~\ref{Lemma1_Proof}}
\label{app: B}
The Laplace transform of interference conditioned on the distance of the hole center to the origin, $\|\nby\|$, is
\begin{align}\nonumber
%\label{LtI asym_001}
&\ncalL_{{I | \|\nby\|}}(s)=
\E\left[\exp \left(-s \sum_{\nbx \in \Phi_\mathrm{2} \cap \nbb^c(\nby,D))}P h_\nbx \|\nbx\|^{-\alpha}\right)\right]\\ \nonumber&=
\E_{\Phi_\mathrm{2}} \left[\prod_{\nbx\in \Phi_\mathrm{2}\cap \nbb^c(\nby,D))}\!\! \E_{h_\nbx} \left[\exp \left(-s P h_\nbx \|\nbx\|^{-\alpha}\right)\right]\right]
\\\nonumber &=\E_{\Phi_\mathrm{2}}\! \left[\prod_{\nbx \in \Phi_\mathrm{2}\cap \nbb^c(\nby,D))}\frac{1}{1+s P \|\nbx\|^{-\alpha}}\right]\\ \nonumber
&\stackrel{\text{(a)}}=\exp\left(\!-\lambda_2\!\int_{\R^2 \setminus \ncalC}\! \frac{1}{1+\frac{\|\nbx\|^{\alpha}}{sP}}\nrmd \nbx\right)
\!\!\\\nonumber &=\exp \left(\!\!-\lambda_2\!\left( \int_{\R^2}\frac{1}{1+\frac{\|\nbx\|^{\alpha}}{sP}}\nrmd \nbx-\!\! \int_{\mathbf{b}(\nby,D)}\!\frac{1}{1+\frac{\|\nbx\|^{\alpha}}{sP}}\nrmd \nbx\right)\right)\\ \nonumber
%&\stackrel{\text{(b)}}=\exp\left(-\pi\lambda_2 \frac{{(sP)}^{2/\alpha}}{\sinc(2/\alpha)}\right) \exp\left(\lambda_2 \int_{\|\nby\|-D}^{\|\nby\|+D} \int_{-\theta(r)}^{\theta(r)} \frac{\nrmd \theta r\nrmd r}{1+\frac{r^{\alpha}}{sP}}\right)\\\nonumber
&\stackrel{\text{(b)}}{=}\exp\left(-\pi\lambda_2 \frac{{( sP)}^{2/\alpha}}{\sinc(2/\alpha)}\right)\\\nonumber &\times\exp\left(2\lambda_2\int_{\|\nby\|-D}^{\|\nby\|+D}\frac{ \arccos\left(\frac{r^2+
\|\nby\|^2-D^2}{2\|\nby\|r}\right)}{1+\frac{r^{\alpha}}{sP}} r \nrmd r\right)
\end{align}
where (a) follows from the expression for the PGFL of a PPP and (b) is derived by the standard machinery, where the integral is first converted form Cartesian to polar coordinates and the closed form expression is then derived by using the properties of the Gamma function (in the same way as step (c) in the proof of Lemma~\ref{thm :lowbound} in Appendix \ref{app: A})~\cite[Appendix B]{dhillon2012modeling}.
The second term follows from the cosine-law: $r^2+\|\nby\|^2- 2 r\|\nby\| \cos\theta(r)=D^2$ (\figref{proof04}). 
By substituting $\lambda(r)= \frac{\lambda_2}{\pi}{\arccos\left(\frac{r^2+\|\nby\|^2-D^2}{2\|\nby\|r}\right)}$, the final expression in equation~(\ref{LtI asym_03}), is derived. 
%An equivalent interpretation of the above transformation is that instead of removing a single hole from interference field, we can capture the impact of asymmetric interference distribution by a non-homogeneous point process of density $\lambda(r)$.

\subsection{Proof of Theorem~\ref{thm :nearhole}}
\label{app: C}
The lower bound on the Laplace transform of interference is
\begin{align}
\label{LtI asym_0001}
&\ncalL_{{I}}(s)\geq\E\left[\exp \left(-s \sum_{\nbx \in \Phi_\mathrm{2}\cap \nbb^c(\nby_1,D)}P h_\nbx \|\nbx\|^{-\alpha}\right)\right]\\\nonumber &=\int_D^{\infty}\ncalL_{{I|\|\nby_1\|}}(s; \lambda, D)f_{V_1}(v_1)\nrmd v_1\\ \nonumber
&\stackrel{\text{(a)}}{=}\exp\left(-\pi\lambda_2 \frac{{( sP)}^{2/\alpha}}{\sinc(2/\alpha)}\right)\times  \\\nonumber
&\int_D^{\infty}\exp\left(\int_{v_1-D}^{v_1+D}2\lambda_2\frac{ \arccos\left(\frac{r^2+
v_1^2-D^2}{2v_1r}\right)}{1+\frac{r^{\alpha}}{sP}} r \nrmd r\right)\\\nonumber &\times 2\pi \lambda_1 v_1\exp(-\pi\lambda_1 (v_1^2-D^2))\nrmd v_1
\\ \nonumber
&=\exp\left(-\pi\lambda_2 \frac{{( sP)}^{2/\alpha}}{\sinc(2/\alpha)}\right)\\\nonumber &\times  \int_D^{\infty}\exp\left(g(v_1)\right)2\pi \lambda_1 v_1\exp(-\pi\lambda_1 (v_1^2-D^2))\nrmd v_1
\end{align}
where ${\mathbf{b}(\nby_1,D)}$ denotes the exclusion zone centered at $\nby_1$ with radius $D$, 
and (a) follows by substituting the conditional Laplace transform expression from Lemma \ref{Lemma1_Proof}, and the PDF of $V_1$ from \eqref{Eq: FRcrofton}. 
Further, $g(v_1)=\int_{{v_1-D}}^{{v_1+D}}{\arccos\left(\frac{r^2+v_1^2-D^2}{2v_1r}\right)}\frac{2\lambda_2}{1+\frac{r^\alpha}{Ps}} r\mathrm{d}r$. 

\subsection{Proof of Theorem~\ref{IccThm3}}
\label{app: D}
By definition, the Laplace transform of the PHP is
\begin{align*}
&\ncalL_I(s)\stackrel{\text{(a)}}{=}\E \left[\exp\left(-s \sum_{\nbx\in \Phi_2 \cap \Xi_D^c} P h_\nbx \|\nbx\|^{-\alpha}\right)\right]\\
&\stackrel{\text{(b)}}{=} \E_{\Phi_1} \left[\exp\left(-\lambda_2 \left(\int_{\R^2 }\frac{\mathrm{d}\nbx}{1+{\frac{\|\nbx\|^{\alpha}}{sP}}}- \int_{\Xi_D}\frac{\mathrm{d}\nbx}{1+{\frac{\|\nbx\|^{\alpha}}{sP}}}\right)\right)\right]
\end{align*}
where $\Xi_D$ in (a) is $\triangleq  \bigcup_{\nby \in \Phi_1} \mathbf{b}(\nby,D)$ as defined in \eqref{eq:XiD}, (b) follows by taking expectations over channel gains $h_\nbx \sim \exp(1)$ and the PPP $\Phi_2$ given $\Xi_D$, where we use the PGFL of a PPP to take expectation over $\Phi_2$. Note the integral over $\Xi_D$ is not easy to compute due to the possible overlaps in the holes. Therefore, to derive the bound, we use
\begin{align*}
\int_{\Xi_D}\frac{\mathrm{d}\nbx}{1+{\frac{\|\nbx\|^{\alpha}}{sP}}} \leq \sum_{\nby \in \Phi_1} \int_{\nbb(\nby, D)} \frac{\mathrm{d}\nbx}{1+{\frac{\|\nbx\|^{\alpha}}{sP}}},
\end{align*}
which follows by ignoring the effect of overlaps. Substituting this back in the expression of $\ncalL_I(s)$; solving the first integral as done in Lemma~\ref{Lemma1_Proof}; and using the result of Lemma~\ref{Lemma1_Proof} to handle the integral over $\nbb(\nby, D)$, we get 
\begin{align*}
&\ncalL_I(s) \leq \exp\left(-\pi\lambda_2 \frac{{(sP)}^{2/\alpha}}{\mathrm{sinc}(2/\alpha)}\right)\\\nonumber &\times  \E_{\Phi_1}\!\! \left[\prod_{\nby\in\Phi_1} \!\!\exp\! \left(\!\!2\lambda_2\int_{\|\nby\|-D}^{\|\nby\|+D}\frac{ \arccos\left(\frac{r^2+
\|\nby\|^2-D^2}{2\|\nby\|r}\right)}{1+\frac{r^{\alpha}}{sP}} r \nrmd r\right)
\right]\\
&\stackrel{\text{(a)}}{=} \exp\left(-\pi\lambda_2 \frac{{(sP)}^{2/\alpha}}{\mathrm{sinc}(2/\alpha)}\right)\\\nonumber &\times
\exp\left[-2\pi\lambda_1 \int_{D}^{\infty}\left(1-\exp\left(f(v) \right)\right)v\mathrm{d}v\right], \nonumber
\end{align*}
where the second term in (a) follows from the PGFL of a PPP, and then by substituting $\|\nby\|=v$ and $f(v)=\int_{{v-D}}^{{v+D}}{\arccos\left(\frac{r^2+v^2-D^2}{2vr}\right)}\frac{2\lambda_2}{1+\frac{r^\alpha}{Ps}} r\mathrm{d}r$. Since by definition of the typical point in this case, there are no points of $\Phi_1$ in $\nbb(0,D)$, the lower bound of integral in the above expression is $D$.

\subsection{Proof of Proposition~\ref{RatioBnd}}
\label{app: E}
Denote the interference powers used for deriving the lower and upper bounds on the Laplace transform of interference in Theorems \ref{thm :nearhole} and \ref{IccThm3} by  $I_\nrml = \sum_{\nbx \in \Phi_2 \cap \nbb^c(\nby_1,D)} P h_{\nbx} \|\nbx\|^{-\alpha}$, and $I_\nrmu = \sum_{\nbx \in \Phi_2} P h_{\nbx} \|\nbx\|^{-\alpha}-\sum_{\nby \in \Phi_1}\sum_{\nbx \in \Phi_2\cap \nbb(\nby,D)} P h_{\nbx} \|\nbx\|^{-\alpha}$, respectively. 
Here, $\nby_1$ denotes the location of the closest point of $\Phi_1$ to the origin. Using this notation, the ratio of the upper and lower bounds is
\begin{align*}
\frac{{\ncalL_\nrmu (s)}}{{\ncalL_\nrml (s)}}&=\frac{\nbbE\left[e^{-sI_\nrmu}\right]}{\nbbE\left[e^{-sI_\nrml}\right]}\\
&\stackrel{\text{(a)}}{\leq} \nbbE\left[e^{-sI_\nrmu}\right] \nbbE \!\left[ \frac{1}{e^{-sI_\nrml} } \right]
\stackrel{\text{(b)}}{\approx} \nbbE \left[e^{-s(I_\nrmu-I_\nrml)} \right],
\end{align*}
where (a) follows from the Jensen's inequality, and (b) is an approximation because $I_\nrmu$ and $I_\nrml$ are not truly independent. We will numerically show that the resulting expression provides a tight approximation. 
%where (a) is an approximation because $I_\nrmu-I_\nrml$ and $I_\nrml$ are dependent. 
As it is clear from the proof of Theorem~\ref{IccThm3}, $I_\nrmu$ is the effective interference from $\Phi_2$ when holes corresponding to $\Phi_1$ are carved out individually without worrying about the overlaps. In other words, some points of $\Phi_2$ may be {\em virtually} removed multiple times, thus leading to an upper bound on the Laplace transform. This means that $I_\nrmu-I_\nrml$ term in the above expression can be interpreted as the effective {\em interference power} removed by all the holes except the closest hole from the homogeneous PPP $\Phi_2$, where again the overlap among the holes is ignored.  On the same lines as the proof of Theorem~\ref{IccThm3}, the term $\nbbE \left[e^{-s(I_\nrmu-I_\nrml)} \right]$ can be evaluated as
\begin{align*}
&\nbbE_{\Phi_1|V_1}\exp\!\left(\!{2\lambda_2}\!\left(\sum_{\nby\in \Phi_1/\nby_1}\!\!\!\int_{{\|\nby\|-D}}^{{\|\nby\|+D}}\!\frac{\arccos(\frac{r^2+\|\nby\|^2-D^2}{2\|\nby\|r})}{1+r^\alpha/s} r\mathrm{d}r\!\right)\!\!\right)\\
&\stackrel{\text{(a)}}{=}
\exp\left[-2\pi\lambda_1 \int_{v_1}^{\infty}\left(1-\exp\left(f(v) \right)\right)v\mathrm{d}v\right],
\end{align*}
where (a) is obtained from PGFL of a PPP, and $f(v)=\int_{{v-D}}^{{v+D}}{\arccos\left(\frac{r^2+v^2-D^2}{2vr}\right)}\frac{2\lambda_2}{1+\frac{r^\alpha}{Ps}} r\mathrm{d}r$. Note that $V_1=\|\nby_1\|$ is the distance of the closest point of $\Phi_1$ from the origin. Deconditioning over the distance $V_1$ using the distribution given by \eqref{Eq: FRcrofton} completes the proof.

\subsection{Proof of Theorem~\ref{Lemma2_nProof}}
\label{app: F}
In order to derive the lower bound on the Laplace transform of interference, the interference is overestimated by considering only two holes that are closest to the typical node. The setup is illustrated in Fig.~\ref{pdf_w}. The idea is to first derive the Laplace transform conditioned on $V_1=\|\nby_1\|$, $V_2=\|\nby_2\|$, and $\phi$, which is the angle between two holes. Deconditioning on these random variables will yield the final result. The details are as follows:  
\begin{align}
\label{condLapl9} 
&\ncalL_I(s)\stackrel{\text{(a)}}{\geq}\E \left[\exp\left(-s\!\!\!\! \sum_{\substack{\nbx\in {\Phi_\mathrm{2}\cap \Xi_C^c}}}\!\! P h_\nbx \|\nbx\|^{-\alpha}\right)\right] \\ \nonumber
&=\E \left[\E_{\Phi_2} \left[\prod_{\substack{\nbx\in {\Phi_\mathrm{2}\cap \Xi_C^c}}} \E_{h_\nbx}\left[\exp\left(-s  P h_\nbx \|\nbx\|^{-\alpha}\right)\right]\right]\right]\\ \nonumber
&\stackrel{\text{(b)}}{=}\E \left[\E_{\Phi_2} \left[\prod_{\nbx\in {\Phi_\mathrm{2}\cap \Xi_C^c}} \frac{{\mathrm{d}\nbx}}{1+s  P \|\nbx\|^{-\alpha}}\right]\right]\\ \nonumber
&\stackrel{\text{(c)}}{=} \E \left[\exp\left(-\lambda_2 \int_{\R^2\setminus \Xi_C }\frac{\mathrm{d}\nbx}{1+{\frac{\|\nbx\|^{\alpha}}{sP}}}\right)\right]\\ \nonumber
&= \E \left[\exp\left(-\lambda_2 \left(\int_{\R^2 }\!\!\frac{\mathrm{d}\nbx}{1+{\frac{\|\nbx\|^{\alpha}}{sP}}}- \int_{\Xi_C}\frac{\mathrm{d}\nbx}{1+{\frac{\|\nbx\|^{\alpha}}{sP}}}\right)\right)\right]\\ \nonumber
&= \E \Bigg[\exp\left(-\pi\lambda_2 \frac{{(sP)}^{2/\alpha}}{\mathrm{sinc}(2/\alpha)}\right) 
\underbrace{\exp\left(\lambda_2 \int_{\ncalC_1}\frac{\mathrm{d}\nbx}{1+{\frac{\|\nbx\|^{\alpha}}{sP}}}\right)}_{\rm closest\ hole} 
\\ \nonumber
&\underbrace{\exp\left(\lambda_2 \int_{\ncalC_2}\frac{\mathrm{d}\nbx}{1+{\frac{\|\nbx\|^{\alpha}}{sP}}}\right)}_{\rm second\ closest\ hole} 
\underbrace{\exp\left(-\lambda_2 \int_{\ncalC_1\bigcap\ncalC_2}\frac{\mathrm{d}\nbx}{1+{\frac{\|\nbx\|^{\alpha}}{sP}}}\right)}_{\rm intersection\ of\ the\ two\ holes}
\Bigg]\\\nonumber 
&\stackrel{\text{(d)}}{=}\exp\left(-\pi\lambda_2 \frac{{(sP)}^{2/\alpha}}{\mathrm{sinc}(2/\alpha)}\right) \times\\ \nonumber
&\bigg(\frac{1}{2\pi } \int_D^{\infty}\!\!\int_{v_1}^{\infty}\!\!\!\!\int_{-\pi}^{\pi}\!\!\exp\left(\int_{v_1-D}^{v_1+D}\!\!\frac{2\pi \lambda_{c1}(r)}{1+\frac{r^{\alpha}}{sP}} r \nrmd r\right)
 \\\nonumber & \exp\left(\int_{v_2-D}^{v_2+D}\!\!\frac{2\pi \lambda_{c2}(r)}{1+\frac{r^{\alpha}}{sP}} r \nrmd r\right)
\exp\left(-\lambda_2\underbrace{\int_{\ncalC_1\bigcap\ncalC_2}\frac{\mathrm{d}\nbx}{1+{\frac{\|\nbx\|^{\alpha}}{sP}}}}_{\text{$\ncalB(v_1,v_2, \phi)$}}\right)
\\ \nonumber
&\times f_{V_1V_2}({v_1,v_2})\nrmd \phi \nrmd v_2 \nrmd v_1
%+\\\nonumber &\mathbf{1} {\{\phi > \|\hat{\phi}\|}\} \int_D^{\infty}\!\!\int_{v_1}^{\infty}\!\!\!\!\exp\left(\int_{v_1-D}^{v_1+D}\!\!\frac{2\pi \lambda_{c1}(r)}{1+\frac{r^{\alpha}}{sP}} r \nrmd r\right)
%\exp\left(\int_{v_2-D}^{v_2+D}\!\!\frac{2\pi \lambda_{c2}(r)}{1+\frac{r^{\alpha}}{sP}} r \nrmd r\right)
%\\\nonumber &
%f_{V_1V_2}({v_1,v_2})\nrmd \phi \nrmd v_2 \nrmd v_1
\bigg),
\end{align}
where $\Xi_C$ in step (a) is $\Xi_C=\ncalC_1 \cup \ncalC_2$ with $\ncalC_1=\nbb(\nby_1,D)$ and $\ncalC_2=\nbb(\nby_2,D)$, and $\lambda_{c1}(r)$ and $\lambda_{c2}(r)$ in the last step are $\lambda_{c1}(r)= \frac{\lambda_2}{\pi}{\arccos\left(\frac{r^2+{v_1}^2-D^2}{2{v_1}r}\right)}$ and $\lambda_{c2}(r)= \frac{\lambda_2}{\pi}{\arccos\left(\frac{r^2+{v_2}^2-D^2}{2{v_2}r}\right)}$. Step (b) follows from  $h_\nbx \sim \exp(1)$, (c) from the PGFL of PPP, and (d) by deconditioning on the distributions of $V_1$, $V_2$, and $\phi$. Note that while the joint distribution of $V_1$ and $V_2$ is given by \eqref{pdf_joint2}, $\phi$ is independent of all other random variables and is uniformly distributed between $-\pi$ and $\pi$. In step (d), the terms corresponding to the closest and second closest holes are derived on the same lines as Lemma~\ref{Lemma1_Proof} (conditioned on $V_1$ and $V_2$). The rest of the proof will focus on evaluating the integral $\ncalB(v_1,v_2, \phi)$ that appears in the term corresponding to the intersection of the two holes. Our first goal is to find the coordinates of the points at which the two circles $\ncalC_1$ and $\ncalC_2$ intersect. 
Without loss of generality, we assume that the centers of the two circles are locates at $\nby_1=(\yuu,0)$ and $\nby_2=(\yuus,\yuut)$ while they are separated by distance $w=\sqrt{v_1^2+v_2^2-2v_1v_2\cos{\phi}}$. Using the equations of the two circles, $(\u-\yuu)^2+{\t}^2=D^2$, and $(\u-\yuus)^2+({\t-\yuut})^2={D^2}$, we obtain 
\begin{align*}
&{-2\u \yuu+\yuu^2=-2\u \yuus + \yuus^2 -2\t \yuut+\yuut^2},
\end{align*}
using which we find $\t$ in terms of other variables. 
Then, substituting $\t$ in one of the equations of the circle, we get a quadratic equation for $\u$. Solving these equations, the coordinates of the intersection points can be found to be~\cite{roe1993elementary}
% expression into the equations of the circles gives a quadratic equation in respect of  horizontal coordinates which The vertical coordinates can be obtained by substituting horizontal coordinates back in the expression given in respect of horizontal coordinates. 
\begin{align}\nonumber
(\hat{u}_{1},\hat{t}_{1})=&\Bigg(\frac{1}{2}\left(\yuu+\yuus\right)+ \frac{1}{2}\sqrt{\frac{4D^2}{w^2}-1}\:\yuut\\ \nonumber
&,\frac{\yuut}{2}+\frac{1}{2}\sqrt{\frac{4D^2}{w^2}-1}\:\left(\yuu-\yuus\right)\Bigg)
\end{align}
\begin{align}\nonumber
(\hat{u}_{2},\hat{t}_{2})=&\Bigg(\frac{1}{2}\left(\yuu+\yuus\right)- \frac{1}{2}\sqrt{\frac{4D^2}{w^2}-1}\:\yuut\\ \nonumber
&,\frac{\yuut}{2}-\frac{1}{2}\sqrt{\frac{4D^2}{w^2}-1}\:\left(\yuu-\yuus\right)\Bigg).
\end{align}
Substituting $(\yuus,\yuut)=(v_2\cos\phi, v_2\sin\phi)$, $(\yuu,\yut)=(v_1,0)$ in these expressions, we get
\begin{align}\nonumber
(\hat{u}_{1},\hat{t}_{1})=&\frac{1}{2}\Bigg(v_1+v_2\cos\phi + \sqrt{\frac{4D^2}{w^2}-1}\: v_2\sin\phi\\ \nonumber
&,v_2\sin\phi + \sqrt{\frac{4D^2}{w^2}-1}\:\left(v_1-v_2\cos\phi\right)  \Bigg)\\ \nonumber
(\hat{u}_{2},\hat{t}_{2})=&\frac{1}{2}\Bigg(v_1+v_2\cos\phi - \sqrt{\frac{4D^2}{w^2}-1}\: v_2\sin\phi\\ \nonumber
&,v_2\sin\phi - \sqrt{\frac{4D^2}{w^2}-1}\:\left(v_1-v_2\cos\phi\right)  \Bigg).
\end{align}
%\begin{align}
%&\ncalC_{1}^{+},\ncalC_{1}^{-}=\pm \sqrt{D^2-(s-y_{1s})^2}\\ \nonumber
%&\ncalC_{2}^{+},\ncalC_{2}^{-}=\yuut \pm \sqrt{D^2-(s-y_{2s})^2}
%\end{align}
%The intersection region is symmetric about the line $t=\frac{\yuut}{y_{2s}-y_{1s}}(s-y_{1s})$ and hence 
Note that the overlap between the circles will happen only when the distance between their centers is smaller than $2D$, i.e., $w \leq 2D$. For a given $v_1$, and $v_2 \geq v_1$, it can be easily deduced that the overlap occurs only when $\phi \leq \hat{\phi}$, where $\hat{\phi} = \arccos \left( \frac{v_1^2 + v_2^2 - 4D^2}{2v_1 v_2} \right)$. The integral in the term corresponding to the intersection of the two circles, $\ncalB(v_1,v_2, \phi)$,  can now we derived as 
\begin{align}\nonumber
&\ncalB(v_1,v_2, \phi)=\int_{\ncalC_1\bigcap\ncalC_2}\frac{\mathrm{d}\nbx}{1+{\frac{\|\nbx\|^{\alpha}}{sP}}}\\ \label{ovl_equav1v2ph}
&=\left\{\begin{matrix}
\int_{\hat{u}_{2}}^{\hat{u}_{1}} \int_{\yuut - \sqrt{D^2-(\u-\yuus)^2}}^{\sqrt{D^2-(\u-\yuu)^2}}\frac{\mathrm{d}\t\mathrm{d}\u}{1+{\frac{{(\u^2+\t^2)}^{\frac{\alpha}{2}}}{sP}}} & 0\leq \phi< \hat{\phi} \\ 
%\int_{\hat{t}_{1}}^{\hat{t}_{2}} \int_{\yuus - \sqrt{D^2-(\t-\yuut)^2}}^{\yuu+\sqrt{D^2-(\t)^2}}\frac{\mathrm{d}\t\mathrm{d}\u}{1+{\frac{{(\u^2+\t^2)}^{\frac{\alpha}{2}}}{sP}}} &  \phi=0 \\ 
\int_{\hat{u}_{1}}^{\hat{u}_{2}} \int_{-\sqrt{D^2-(\u-\yuu)^2}}^{\yuut + \sqrt{D^2-(\u-\yuus)^2}}\frac{\mathrm{d}\t\mathrm{d}\u}{1+{\frac{{(\u^2+\t^2)}^{\frac{\alpha}{2}}}{sP}}} & -\hat{\phi}< \phi< 0\\ 0 & |\phi| \geq \hat{\phi}
\end{matrix}\right.,
\end{align}
where $\|\nbx\|=\sqrt{{\u}^2+{\t}^2}$, $\hat{\phi} = \arccos \left( \frac{v_1^2 + v_2^2 - 4D^2}{2v_1 v_2} \right)$. This completes the proof.

\subsection{Proof of Theorem~\ref{Lemma2_kClosestBnd}}
\label{app: H}
We consider $k$ closest holes to the typical point of a PHP. The setup is illustrated in Fig.~\ref{ovlp_3}. Denoting the locations of the holes by $\nby_1, ...,  \nby_i, ..., \nby_k$,  
the interference field in this case is modeled by $\Omega=\Phi_2 \cap \left\{{{\ncalC_1}\cup {\left\{ \cup_{i=2}^k\nbd(\nby_i,D)\right\}}}\right\}^c$, where  $\Omega \supset \Psi$ and $\nbd(\nby_i,D)$ is defined in \eqref{eq:sect2:dy-D}.
%
%\begin{align*}
%\left(\Phi_2 \cap \left\{{{\ncalC_1}\cup {\left\{ \cup_{i=2}^k\nbd(\nby_i,D)\right\}}}\right\}^c\right) \supset \Psi
%\end{align*} 
This approach overestimates the interference power and hence leads to a lower bound on the Laplace transform of interference. Let $I=\sum_{\substack{\nbx\in {\Omega}}}Ph_\nbx\|\nbx\|^{-\alpha}$, the Laplace transform of interference conditioned on $V_1=\|\nby_1\|, V_2=\|\nby_2\|, ..., V_k=\|\nby_k\|$ is
\begin{align}  \notag
&\ncalL_{{I|\|\nby_1\|,...,\|\nby_k\|}}(s)=
\E\left[\exp \left(-s \sum_{\nbx \in \Omega}P h_\nbx \|\nbx\|^{-\alpha}\right)\right]
\\ \nonumber
&\stackrel{\text{(a)}}{=}\E_{\Omega}\! \left[\prod_{\nbx \in \Omega}\frac{1}{1+s P \|\nbx\|^{-\alpha}}\right]\\ \notag
&\!\!  \stackrel{\text{(b)}}{=}\exp \Bigg(\!\!-\lambda_2\!\Bigg( \int_{\R^2}\frac{1}{1+\frac{\|\nbx\|^{\alpha}}{sP}}\nrmd \nbx-\!\! \int_{\mathbf{b}(\nby_1,D)}\!\frac{1}{1+\frac{\|\nbx\|^{\alpha}}{sP}}\nrmd \nbx \\ \nonumber
&-\sum_{i=2}^k  \int_{\nbd(\nby_i,D)} \frac{1}{1+\frac{\|\nbx\|^{\alpha}}{sP}}\nrmd \nbx  \Bigg)\Bigg)\\ \notag
%&\stackrel{\text{(b)}}=\exp\left(-\pi\lambda_2 \frac{{(sP)}^{2/\alpha}}{\sinc(2/\alpha)}\right) \exp\left(\lambda_2 \int_{\|\nby\|-D}^{\|\nby\|+D} \int_{-\theta(r)}^{\theta(r)} \frac{\nrmd \theta r\nrmd r}{1+\frac{r^{\alpha}}{sP}}\right)\\\nonumber
&\stackrel{\text{(c)}}{=}\exp\left(-\pi\lambda_2 \frac{{( sP)}^{2/\alpha}}{\sinc(2/\alpha)}\right)\\ \nonumber
&\times\exp\left(2\lambda_2\int_{\|\nby_1\|-D}^{\|\nby_1\|+D}\frac{ \arccos\left(\frac{r^2+
\|\nby_1\|^2-D^2}{2\|\nby_1\|r}\right)}{1+\frac{r^{\alpha}}{sP}} r \nrmd r\right)\times\\ \nonumber
& \exp\!\left(\!2 \lambda_2\!\!  \int_{\max({\|\nby_i\|-D,\|\nby_{i-1}\|+D})}^{\|\nby_i\|+D} \frac{ \arccos\left(\frac{r^2+
\|\nby_i\|^2-D^2}{2\|\nby_i\|r}\right)}{1+\frac{r^{\alpha}}{sP}} r \nrmd r \!\right)
\end{align}
where (a) follows from $h_{\nbx}\sim \exp(1)$, (b) from  PGFL of PPP along with the fact that the points of the baseline PPP $\Phi_2$ that are located in the closest hole $\nbb(\nby_1,D)$ and the sets $\cup_{i=2}^k \nbd(\nby_i, D)$ 
should be removed, and (c) from the cosine-law by using the same argument applied in the proof of Lemma \ref{Lemma1_Proof}. Finally, deconditioning the resulting expression with respect to the distances of the centers of the $k$ holes from the typical point, $V_1, V_2, ..., V_k$, with joint distribution given by 
\begin{align*}
&f_{V_1V_2..V_k}({v_1,v_2,..,v_k})=(2\pi\lambda_1)^k v_1v_2...v_k\times\\\nonumber
&  \exp(-\pi\lambda_1 (v_k^2-D^2)), \quad {D<v_1<v_2<...<v_k},
\end{align*}
completes the proof. Here, the joint PDF $f_{V_1V_2..V_k}(.)$ is derived by using the same argument as in equation~(\ref{pdf_joint2})~\cite{haenggi2005distances}.

\subsection{Proof of Proposition~\ref{IccThm4}}
\label{app: G} 
 \begin{figure}[t!]
  \centering{
              \includegraphics[width=.65\linewidth]{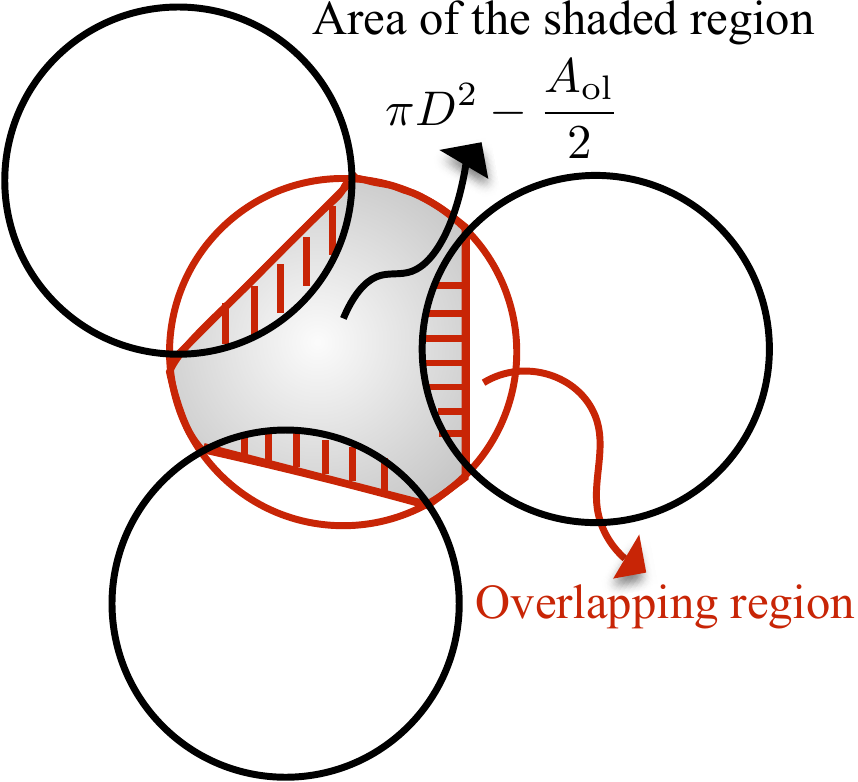}
              \caption{Illustration of the overlap effect for a typical hole in a PHP. This is used in Appendix~\ref{app: G}.}
\label{ovlp_appr}}
          \end{figure}
%While deriving the upper bound given by Theorem~\ref{IccThm3}, the holes were carved out from the baseline process $\Phi_2$ without caring for the overlaps. The average number of nodes removed per hole can be estimated as

Consider the illustration of overlapping circles shown in Fig.~\ref{ovlp_appr}. While defining a PHP, points of the baseline PPP $\Phi_2$ lying inside these circles are removed. In other words, the ``union'' of such overlapping circles is removed. However, in our upper bound given by Theorem~\ref{IccThm3}, we removed each such circle separately of the others leading to the removal of the overlapping parts multiple times. Therefore, while the average number of points removed per unit area of these circles was $\lambda_2$ in a PHP, the average number of points removed per unit area of such circles in Theorem~\ref{IccThm3},  denoted by $\lambda_{\rm eff}$, is $\lambda_{\rm eff} > \lambda_2$. To maintain tractability later in the proof, we confine to the pairwise overlaps, meaning the total overlap area for a circle, denoted by $A_\mathrm{ol}$, is the sum of pairwise overlap areas of this circle with the other circles. Now to estimate $\lambda_{\rm eff}$, consider the shaded region in Fig.~\ref{ovlp_appr} with area $\pi D^2- \frac{A_\mathrm{ol}}{2}$, which represents the {\em effective} contribution of each hole in a PHP. Due to overlaps, the total average number of points removed under Theorem~\ref{IccThm3} from this region is $\lambda_2 \pi D^2$. Using this fact, $\lambda_{\rm eff}$ can be estimated as 
\begin{align}\nonumber
\lambda_\mathrm{eff}\times \left({{\pi}D^2}-\frac{A_\mathrm{ol}}{2}\right)=\lambda_2{{\pi}D^2}\Rightarrow \lambda_\mathrm{eff}=\frac{\lambda_2}{\left(1-\frac{A_\mathrm{ol}}{2{{\pi}D^2}}\right)}.
\end{align}
%equal share of each pair of overlapping areas.
The above expression shows that $\frac{1}{(1-{\frac{A_\mathrm{ol}}{2{{\pi}D^2}}})}$ times more points are removed in the upper bound given by Theorem~\ref{IccThm3}. To compensate for this effect, we first proceed as in Theorem~\ref{IccThm3} and then rescale the second term (that captures the effect of removing holes) as follows:
\begin{align}\nonumber%\label{rescale_ovl}
&\ncalL_I(s)\stackrel{\text{(a)}}{\simeq} \E_{\Phi_1} \Bigg[\exp\Bigg(-\lambda_2 \bigg(\int_{\R^2 }\frac{\mathrm{d}\nbx}{1+{\frac{\|\nbx\|^{\alpha}}{sP}}} \\ \nonumber
&-\left({1-\frac{A_\mathrm{ol}}{2{{\pi}D^2}}}\right)\sum_{\nby \in \Phi_1} \int_{\nbb(\nby, D)} \frac{\mathrm{d}\nbx}{1+{\frac{\|\nbx\|^{\alpha}}{sP}}}\bigg)\Bigg)\Bigg]\\ \nonumber
&=
\exp\left(\!-\pi\lambda_2 \frac{{(sP)}^{2/\alpha}}{\mathrm{sinc}(2/\alpha)}\right)\!\!\times \E_{\Phi_1}\! \Bigg[\!\prod_{\nby\in\Phi_1} \!\!\exp\! \Bigg(\!\!{2\lambda_2\left({1-\frac{A_\mathrm{ol}}{2{{\pi}D^2}}}\right)}\\ \nonumber
&\times\int_{\|\nby\|-D}^{\|\nby\|+D}\frac{ \arccos\left(\frac{r^2+
\|\nby\|^2-D^2}{2\|\nby\|r}\right)}{1+\frac{r^{\alpha}}{sP}} r \nrmd r\Bigg)
\Bigg]\\ \nonumber
&\stackrel{\text{(b)}}{=} \exp\left(-\pi\lambda_2 \frac{{(sP)}^{2/\alpha}}{\mathrm{sinc}(2/\alpha)}\right)\times \\ \nonumber
&
\exp\left[-2\pi\lambda_1 \int_{D}^{\infty}\left(1-\exp\left({f(v)}\left({1-\frac{A_\mathrm{ol}}{2{{\pi}D^2}}}\right) \right)\right)v\mathrm{d}v\right]
\end{align} 
where (a) follows from $h_\nbx \sim \exp(1)$ and the PGFL of a PPP. The second term is rescaled as discussed above. Further, the second term in (b) follows from the PGFL of a PPP, and then by substituting $\|\nby\|=v$ and $f(v)=\int_{{v-D}}^{{v+D}}{\arccos\left(\frac{r^2+v^2-D^2}{2vr}\right)}\frac{2\lambda_2}{1+\frac{r^\alpha}{Ps}} r\mathrm{d}r$. 

Now to determine the total average {\em pairwise} overlapping area $A_{\rm ol}$, consider a circle of interest $\mathbf{b}(\nby,D)$. Note that only the circles with centers located inside the region  $\mathbf{b}(\nby,2D)$ will overlap with this circle. Denote the number of circles in this region (besides the circle of interest) by $K$, where $K$ has Poisson distribution with mean $\lambda_1 4{\pi}D^2$. For one of these circles, say $\mathbf{b}(\nby',D)$, the area of overlapping region with the circle of interest is
%Now, we derive the total {\em pairwise} overlapping area between the hole of interest $\mathbf{b}(\nby,D)$ and $k$ holes inside the region $\mathbf{b}(\nby,2D)$ that may overlap with the hole of interest.
%Hence, $k$ has Poisson distribution with mean $\lambda_p4{\pi}D^2$. 
%Overlap for each hole of interest centered at $\nby\in \Phi_1$ happens when the distance $\hat{z}=\|\nby-\nby'\|$ between the hole of interest and every other holes located at ${\nby}'\in \Phi_1 \cap \mathbf{b}(\nby,2D)$ is less than $2D$.
%The overlapping area of  $A(\hat{z})$ between the two holes can be expressed as
\begin{align}\nonumber
A(\hat{z})=2D^2\arccos(\frac{\hat{z}}{2D})-{\hat{z}D}\sqrt{1-(\frac{\hat{z}}{2D})^2},
\end{align}
where $\hat{z}=\|\nby-\nby'\|$. Now conditioned on $K=k$, the $k$ circles are independent and uniformly distributed over $\mathbf{b}(\nby,2D)$. Hence, the PDF of distance $\hat{z}$ is $f_{\hat{Z}}(\hat{z})=\frac{2\hat{z}}{4D^2}$.
Then the average area of overlapping region between the circle of interest and another circle is 
\begin{align}\nonumber
%\label{Aavg_condiK1}
&\bar{A}=\nbbE_{\hat{Z}}\left[A(\hat{z})\right]\\ \nonumber
&= \left(\int_{0}^{2D}
\!\left(2D^2\arccos(\frac{\hat{z}}{2D})-{\hat{z}D}\sqrt{1-(\frac{\hat{z}}{2D})^2}\right)f_{\hat{Z}}(\hat{z})\mathrm{d}\hat{z}\right)\\ \nonumber
&= \frac{\pi D^2}{4}.
\end{align}
%where $\bar{A}$ is the average overlapping area for one of the holes which can intersect with the hole of interest.
%To derive the new approximation for the coverage probability of a typical user in a PHP, pairwise average overlapping area is assumed that maximum two overlapping holes cover each region. 
%$k$ overlapping holes inside the region has no intersection. This assumption simplifies 
Now the pairwise overlap area for the circle of interest with $k$ other circles is $k\bar{A}$. Since there are on an average $\lambda_1 4 \pi D^2$ circles that overlap with the circle of interest, the average pairwise overlap area is $\lambda_1 4 \pi D^2 \bar{A} = \lambda_1{\pi}^2D^4$. Since the maximum overlap is bounded above by $\pi D^2$, we get $A_{\mathrm{ol}}=\min({\lambda_1{\pi}^2D^4},{\pi}D^2)$. This completes the proof.
\bibliographystyle{IEEEtran}
\bibliography{Draft-00JnPHP_24-arXiv_double.bbl} 
%\bibliography{ref,ref-Dhillon}  
               
\end{document}